\documentclass[a4paper,onecolumn,11pt,accepted=2026-10-20]{quantumarticle}
\pdfoutput=1
\usepackage[utf8]{inputenc}
\usepackage[english]{babel}
\usepackage[T1]{fontenc}
\usepackage{amsmath}
\usepackage{hyperref}
\usepackage{doi}
\usepackage{subfigure}
\usepackage{diagbox}
\usepackage{arydshln}
\usepackage{qtree}
\usepackage{algpseudocode}

\usepackage{physics}
\usepackage{amsfonts}
\usepackage{graphicx}  
\usepackage{qcircuit}
\usepackage{amsthm,scalerel}
\usepackage{amssymb}
\usepackage{algorithm}
\usepackage{algpseudocode}
\usepackage{caption}

\newtheorem{theorem}{Theorem}[section]
\newtheorem{corollary}{Corollary}

\newtheorem{Def}{Definition}

\theoremstyle{definition}
\newtheorem{Lem}{Lemma}
\newtheorem{problem}{Open Problem}

\newcommand*{\bordl}{\multicolumn{1}{|c}{}}
\newcommand*{\bordr}{\multicolumn{1}{c|}{}}

\newcommand{\etal}{\textit{et~al.}}

\newcommand{\CH}{\mathcal{CH}}

\def\lc{\left\lfloor}   
\def\rc{\right\rfloor}

\begin{document}

\title{Affine Equivalence in the Clifford Hierarchy}
\author{Jonas T. Anderson}
\email{jonas.tyler.anderson@gmail.com}
\affiliation{Northrop Grumman}
\orcid{0000-0002-3230-6712}
 
\author{Andrew Connelly}
\email{drew.j.connelly@gmail.com}
\affiliation{Northrop Grumman}
\orcid{0000-0002-2596-3706}

\author{Mason Lovejoy-Johnson}
\email{masonlj@utexas.edu}
\affiliation{UT Austin College of Natural Sciences}

\maketitle
\begin{abstract}
    In this paper we prove a collection of results on the structure of permutations in the Clifford Hierarchy. First, we leverage results from the cryptography literature on affine equivalence classes of 4-bit permutations which we use to find all 4-qubit permutations in the Clifford Hierarchy. We then use the classification of 4-qubit permutations and previous results on the structure of diagonal gates in the Clifford Hierarchy to prove that all 4-qubit gates in the third level of the Clifford Hierarchy are semi-Clifford. We then leverage results from \cite{He2024} to prove that all 5-qubit gates in the third level of the Clifford Hierarchy are semi-Clifford. We introduce the formalism of cycle structures to permutations in the Clifford Hierarchy and prove a general structure theorem about them, and classify many small cycle structures up to affine equivalence. Interestingly, this classification is independent of the number of qubits. Additionally, we show that all 1-wire-mismatch permutations are in the Clifford Hierarchy.        
\end{abstract}

\section{Motivation}

The mathematical structure of the qubit Clifford Hierarchy remains an active area of research \cite{Hu2021ClimbingHierarchy, Anderson2024, Pllaha2020, Rengaswamy2019, He2024, Bengtsson2014}. Understanding this structure places important restrictions on the types of transversal gates available to qubit stabilizer codes \cite{Jochym-OConnor2018DisjointnessGates, Anderson2016}. In addition to the work on the qubit Clifford Hierarchy, recent work has begun to elucidate the structure of the qudit Clifford Hierarchy \cite{deSilva2021EfficientDimensions, deSilva2025}. 

In \cite{Anderson2024} we classified all groups in the Clifford Hierarchy comprised of generalized semi-Clifford elements\footnote{The groups we found were Clifford groups and Dihedral-like groups each having support on distinct qubits. There was one possible exception where additional elements with limited support in both groups could not be directly ruled out. See \cite{Anderson2024} for more details.}. We also proved that a gate is generalized semi-Clifford and in $\CH$ if and only if it is Clifford equivalent to a product of a permutation and a diagonal gate where both the permutation and diagonal gate are in $\CH$. The diagonal gates in $\CH$ where classified in \cite{Cui2017}, but a systematic study of the permutations in $\CH$ has not been carried out\footnote{During our research a systematic study of permutations in the third level of the Clifford Hierarchy was published \cite{He2024}. Our work here is mostly concerned with permutations in the Clifford Hierarchy independent of level and our results are complementary to theirs.} Furthermore, it is conjectured that all gates in $\CH$ are generalized semi-Clifford\cite{Zeng2008}. If true, classifying the permutations in $\CH$ would lead to a complete classification of all gates in $\CH$. 

The $n$-qubit Clifford Hierarchy \cite{Gottesman1999, Zeng2008} is recursively defined as 

\begin{equation}\label{CHdef}
    \mathcal{CH}_{k} \equiv \{U | U P U^\dagger \subseteq \mathcal{CH}_{k-1}, \forall P\in\mathcal{P}_n\}
\end{equation}
with the first level ($k=1$) defined as $\mathcal{CH}_1 \equiv \mathcal{P}_n$, the $n$-qubit Pauli group. $\mathcal{CH}_2$ is the $n$-qubit Clifford group. For $k\ge3$, the elements of $\mathcal{CH}_k$ no longer form a group. 

Our current knowledge of the general structure of the qubit Clifford Hierarchy is summarized in Table~\ref{table:CH}.

\begin{table}[!h]
\centering
\begin{tabular}{l|c|c|c|c}
\diagbox{$n$}{$k$} & $k=1$  & $k=2$  & $k=3$ & $k\ge4$\\
\hline
$n=1$ & $\mathcal{P}$ & $\mathcal{C}$ & $\mathcal{SC}$ & $\mathcal{SC}$ \\
\hline
$n=2$ & $\mathcal{P}$ & $\mathcal{C}$ & $\mathcal{SC}$  & $\mathcal{SC}$ \\
\hline
$n=3$ & $\mathcal{P}$ & $\mathcal{C}$ & $\mathcal{SC}$  & $\mathcal{SC}^+$  \\
\hline
$n=4$ & $\mathcal{P}$ & $\mathcal{C}$ & \textcolor{magenta}{$\mathcal{SC}$}  & $\mathcal{SC}^+$ \\
\hline
$n=5$ & $\mathcal{P}$ & $\mathcal{C}$ &  \textcolor{magenta}{$\mathcal{SC}$}  & $\mathcal{SC}^+$ \\
\hline
$n=6$ & $\mathcal{P}$ & $\mathcal{C}$ & $\mathcal{GSC}^-$  & $\mathcal{SC}^+$ 
\\
\hline
$n\ge 7$ & $\mathcal{P}$ & $\mathcal{C}$ & $\mathcal{GSC}$  & $\mathcal{SC}^+$ 
\\
\hline
\end{tabular}
\caption{Here $k$ indicates the level in the Clifford Hierarchy, $n$ the number of qubits, $\mathcal{P}$ denotes the Pauli group, and $\mathcal{C}$ denotes the Clifford group. Table entries with $\mathcal{SC}$ indicate that the elements are semi-Clifford, $\mathcal{GSC}$ indicates that the elements are generalized semi-Clifford, and $\mathcal{SC}^+$ indicates that gates are known to exist which are not semi-Clifford (though they are generalized semi-Clifford), but it has not been proven that all such elements must be generalized semi-Clifford. We use $\mathcal{GSC}^-$ to indicate cases that have been proven to only contain generalized semi-Clifford elements, but only semi-Clifford elements are currently known. $\mathcal{GSC}$ ($k=3,n\ge7$) indicates that these cases have been proven to be generalized semi-Clifford and elements are known which are not semi-Clifford. Our results ($ \mathcal{GSC}^- \rightarrow \mathcal{SC}$ for $(n,k) =(4,3) \mbox{ and } (5,3)$) are presented in \textcolor{magenta}{magenta}. All other results presented in this table were proven in \cite{Zeng2008, Beigi2010}.}\label{table:CH} 
\end{table}

\section{Affine Equivalence}
In \cite{Zeng2008} it was shown that left or right multiplication by Clifford elements does not change the level of a gate in the Clifford Hierarchy. Therefore, when considering which permutation gates are in $\mathcal{CH}$, it would be useful to split permutation gates into Clifford equivalence classes where Clifford equivalence is defined by equality up to left or right multiplication by Clifford elements. To the best of our knowledge, a theory of Clifford equivalence of permutation gates has not yet been developed; however, the cryptography community has developed a theory of affine equivalence\footnote{In appendix~\ref{affineproof} we discuss possible extensions of affine equivalence.} of permutations \cite{Draper2009, Canniere2007}. 

An $n$-bit affine permutation is a linear permutation followed by a bit shift. For qubits, all linear permutations are generated by combinations of $CNOT$ gates and bit shifts are generated by Pauli $X$ gates. Since $CNOT$s and Pauli $X$s generate all Clifford permutations, affine permutations are precisely Clifford permutations. Furthermore, since affine equivalent permutations preserve the level of a permutation in $\mathcal{CH}$, we can leverage results about affine equivalence to answer questions about the Clifford Hierarchy. 

Two permutations are affine equivalent when they satisfy the following definition:  

\begin{Def}\textbf{Affine equivalent:} $P_1 \approx P_2 \iff P_1 = L P_2 R$ where $L,R$ are affine permutations.
\end{Def}

We list the number of number of affine equivalence classes for up to 5 qubits in the table below:

\begin{table}[h!]
\centering
\begin{tabular}{l|c|c|c|c|c}
qubits ($n$) & $1$ & $2$ &  $3$ & $4$ & $5$ \\
\hline 
AE classes & 1 & 1 & 4 & 302 & $2569966041123963092\approx 2^{61}$
\\
\hline
Permutations & $2^1!$ & $2^2!$ & $2^3!$ & $2^4! \approx 2^{44}$ & $2^5! \approx 2^{118}$ \\
\end{tabular}
\caption{From \cite{Draper2009} the number of affine equivalence classes compared to the total number of permutations for up to $n=5$ qubits. }\label{table:numAE} 
\end{table}

We can see that the number of permutations that must be considered can be significantly reduced by considering a single representative of each affine equivalence classes rather than all permutations. We leverage this fact to classify all permutations in $\mathcal{CH}$ up to 4 qubits and later to classify small cycle structures on any number of qubits.

\section{Permutations on 4 qubits}\label{Perms4Qs}
We leverage the results of \cite{Draper2009, Canniere2007} which lists representatives of all affine equivalence classes of permutations on 4 qubits. Of the 302 classes listed we find 5 are in the Clifford Hierarchy. Contrast this with the three-qubit case where 2 of the 4 affine equivalence classes are in the Clifford Hierarchy. Curiously, the affine equivalence classes which we found to be in the Clifford Hierarchy were the last five of 302 listed in  \cite{Canniere2007}. Their list is ordered by a combination of measures (see \cite{Canniere2007} for details) which determine how far a permutation (or affine equivalence class) is from an affine permutation. The last five, therefore, are the closest to affine (including the class of affine permutations, listed as ``Id.'' below). This is unsurprising, however, we were not able to use these measures to construct a rigorous test to determine if a permutation is in $\mathcal{CH}$. We leave it as an open problem.

\begin{problem}
Is there a simple test to determine whether or not a permutation is in $\mathcal{CH}$? 
\end{problem}

We note that such a test exists for diagonal gates by checking if all the entries on the diagonal are in a ring of $2^r$th root-of-unity phases for some non-negative integer, $r$. This is a simple corollary of the results in \cite{Cui2017}. 

We give a circuit representation of an element in each of the 4-qubit affine equivalence classes in $\mathcal{CH}$ below:

\begin{equation*}
\text{Id.},
\begin{array}{ccc}
\Qcircuit @C=0.5em @R=1.3em {
& \ctrl{1} & \qw \\
& \ctrl{1} & \qw \\
& \ctrl{1}& \qw \\
& \targ& \qw \\
}
\end{array},
\begin{array}{ccc}
\Qcircuit @C=0.5em @R=1.3em {
& \qw & \qw \\
& \ctrl{1} & \qw \\
& \ctrl{1}& \qw \\
& \targ& \qw \\
}
\end{array},
\begin{array}{ccc}
\Qcircuit @C=0.5em @R=1.3em {
& \ctrl{1} & \ctrl{1} & \qw \\
& \ctrl{1} & \ctrl{1} & \qw \\
& \ctrl{1}&  \targ & \qw \\
& \targ & \qw & \qw  \\
}
\end{array},
\begin{array}{ccc}
\Qcircuit @C=0.5em @R=1.3em {
& \ctrl{1} & \qw & \qw \\
& \ctrl{1} & \ctrl{1} & \qw \\
& \targ&  \ctrl{1} & \qw \\
& \qw & \targ & \qw  \\
}
\end{array}.
\end{equation*}
These gates (and all other affine equivalent gates) are in: $\mathcal{CH}_1$, $\mathcal{CH}_4\backslash\mathcal{CH}_3$, $\mathcal{CH}_3\backslash \mathcal{CH}_2$, $\mathcal{CH}_4\backslash \mathcal{CH}_3$, $\mathcal{CH}_4\backslash \mathcal{CH}_3$, respectively.

While the number of affine equivalent classes is more manageable than the number of permutations there are still too many to easily check for $n\ge 5$. 

It is useful to think about the set of permutations for each level of the Clifford Hierarchy. We define the $k$th level of the Permutation Clifford Hierarchy, $\mathcal{PH}_k$, on $n$ qubits as the intersection of the $k$th level of the Clifford Hierarchy and the $2^n \times 2^n$ permutation matrices, $\mathtt{Perm(2^n)}$.  
\begin{Def}$\mathcal{PH}_k \equiv \mathcal{CH}_k \cap \mathtt{Perm(2^n)}$
\end{Def}

We observe that for $n=1,2,3,$ and 4; $\mathcal{PH}_k\backslash\mathcal{PH}_{k-1}$ is non-trivial (contains a non-Identity element) up to level $k=n$ and that $\mathcal{PH}_k\backslash\mathcal{PH}_{k-1}$ is trivial for all $k>n$. This raises an interesting question:
\begin{problem}
For each $n$, what is the maximum level, $k$, such that $\mathcal{PH}_k\backslash\mathcal{PH}_{k-1}$ has non-trivial elements? Is it simply $n$?
\end{problem}

For $n>4$ this problem is open. A brute-force attack on this problem, even using AE classes, would be computationally challenging for $n=5$ qubits (see Table~\ref{table:numAE}). These classes could be sampled potentially yielding a non-trivial lower bound on maximum level $k$. Beyond this, additional structure likely needs to be proven to advance this problem. 

\section{Computational attack on $\mathcal{CH}^{n=4}_3$}\label{sec:atk34}

In this section we show that all gates in $\mathcal{CH}^{n=4}_3$ can be placed into a tractable number of Clifford equivalence classes. We do this by creating a set of Clifford equivalence classes that is larger than $\mathcal{CH}^{n=4}_3$, but contains all of $\mathcal{CH}^{n=4}_3$. Then, by selecting a single member of each class and testing whether or not it is semi-Clifford and in the third level of $\mathcal{CH}$, we can determine if any members of $\mathcal{CH}^{n=4}_3$ are not semi-Clifford.  

From Beigi and Shor \cite{Beigi2010} we have the following Lemma:
\begin{Lem}\label{lem:beigishor}
All elements in $\mathcal{CH}_3$ for any number of qubits, $n$, are generalized semi-Clifford.
\end{Lem}

From \cite{Zeng2008} we know that $\mathcal{CH}^{n=3}_3$ is semi-Clifford. A counterexample by Gottesman and Mochon proves that $\mathcal{CH}^{n\ge7}_3$ contains elements which are generalized semi-Clifford, but not semi-Clifford. Therefore, $\mathcal{CH}^{n=4}_3$ is the smallest $n$ such that the structure of $\mathcal{CH}_3$ is not fully known. 

A generalized semi-Clifford gate on $n$ qubits can be written as $$C_1 \pi d C_2$$ where $C_1, C_2$ are $n$-qubit Clifford gates, $\pi$ is a $2^n \times 2^n$ permutation matrix, and $d$ is an $n$-qubit diagonal gate. From \cite{Anderson2024} (Appendix A) we use the following results (which we refer to as Lemmas in this work):

\begin{Lem}\label{lem:groups1}
    $\pi d \in \mathcal{CH}$ if and only if $\pi \in \mathcal{CH}$ and $d \in \mathcal{CH}$.
\end{Lem}
and 
\begin{Lem}\label{lem:groups2}
   $\pi d \in \mathcal{CH}_k$ only if $\pi \in \mathcal{CH}_k$. 
\end{Lem}

From Lemmas~\ref{lem:groups1} and \ref{lem:groups2} we have the following corollary: 
\begin{corollary}
    $\pi d \in \mathcal{CH}_3$ only if $\pi \in \mathcal{CH}_3$ and $d \in \mathcal{CH}_k$
\end{corollary}
for some finite level $k$. Note that the reason we cannot easily restrict $d$ to $\mathcal{CH}_3$ is that conjugation by a non-Clifford permutation can change the level of a diagonal gate.  

We found and listed all 4-qubit affine equivalence classes of permutations in the Clifford hierarchy in Section~\ref{Perms4Qs}. Recall that only one class was in $\mathcal{CH}_3 \backslash \mathcal{CH}_2$ namely the $CCX_{2,3,4}$ gate\footnote{We do not need to consider permutations in $\mathcal{CH}_2$, the class of affine (Clifford) permutations, since the product of a Clifford gate and a diagonal gate is always semi-Clifford.}. It is straightforward to show that $CCX$ is a semi-Clifford permutation. If this were not true, we would already have a counterexample. Lemma~\ref{lem:beigishor} implies that we only need to consider generalized semi-Clifford gates when searching for counterexamples. Up to multiplication by Clifford gates (which we can ignore since they preserve the level of a gate), a generalized semi-Clifford gate is a product of a permutation gate and a diagonal gate. Furthermore, since all diagonal gates are semi-Clifford, any potential counterexample in $\mathcal{CH}_3$ must be a product of a permutation in $\mathcal{CH}_3 \backslash \mathcal{CH}_2$ and a diagonal gate in $\mathcal{CH} \backslash \mathcal{CH}_2$.\footnote{A further consideration is that the product of these gates must not be diagonalizable by Clifford gates or else it would be semi-Clifford.} 

\begin{Lem}\label{Lem:easy1}
    Diagonal gates in $\mathcal{CH}_k$ can be partitioned into equivalence classes where equivalency is defined by multiplication by diagonal Clifford gates. 
\end{Lem}

\begin{Lem}\label{Lem:easy2}
Conjugation by Clifford permutations on all diagonal gates in $\mathcal{CH}_k$ preserves the size and number of these diagonal equivalence classes.  
\end{Lem}
 
Now we will prove that generalized semi-Clifford gates on any number of qubits and for any level in $\mathcal{CH}$ can be partitioned into equivalence classes which are a product of the affine equivalence classes of permutations discussed above and a diagonal equivalence class whose elements are equivalent up to diagonal Clifford gates.

\begin{theorem}
    Any generalized semi-Clifford in $\mathcal{CH}_k$ is equivalent (up to left/right multiplication by Clifford gates) to some $\Tilde{\pi}\Tilde{d}$ where $\Tilde{\pi}$ is any affine-equivalent member in an affine equivalence class of permutations and $\Tilde{d}$ is any diagonal-Clifford-equivalent member in a diagonal equivalence class.  
\end{theorem}

\begin{proof}
      Recall that any generalized semi-Clifford can be written as $U = C_1 \pi d C_2$. Let $\pi = \phi_1 \pi' \phi_2^{-1}$ where $\phi$ are Clifford (affine) permutations. Here $\pi'$ is (by definition) a member of the same affine equivalence class as $\pi$. Then we have $U = C_1 \phi_1 \pi' \phi_2^{-1} \phi_2 d \phi_2^{-1} C_2$. Now, let $\Tilde{\pi} = (\phi_1 \pi' \phi_2^{-1})$ denote a member of an affine equivalence class of permutations. Note that $\phi_2 d\phi_2^{-1}$ is a diagonal matrix in the same level of $\mathcal{CH}$ as $d$. Let $d= d'd_C$ where $d_C$ denotes a diagonal Clifford gate and $d'$ is a member of the same diagonal Clifford equivalence class as $d$. Combining, we have $U = C_1 \Tilde{\pi} \phi_2 d' \phi_2^{-1}\phi_2 d_C \phi_2^{-1} C_2$. We denote a member of the conjugated diagonal equivalence class as: $\Tilde{d} = \phi_2 d'\phi_2^{-1}$. From Lemmas~\ref{Lem:easy1} and \ref{Lem:easy2} we know that conjugation by $\phi_2$ preserves the diagonal equivalence relation. This implies that iterating over members $d'$ and $\Tilde{d}$ (for fixed $\phi_2$) is equivalent. Finally, we have $U \approx C_1 \Tilde{\pi} \Tilde{d} \phi_2 d_C \phi_2^{-1} C_2  \approx C_L \Tilde{\pi} \Tilde{d} C_R \approx \Tilde{\pi} \Tilde{d}$ as claimed. Where we have used ``$\approx$'' to indicate Clifford equivalence. 
\end{proof}

\begin{corollary}
    By checking one member of each equivalence class $\Tilde{\pi} \Tilde{d}$ we can infer the level of all members of the class and whether or not they are semi-Clifford as these properties are not changed by left or right multiplication by Clifford gates.  
\end{corollary}

At this point if we knew which equivalence classes, $\Tilde{\pi} \Tilde{d}$, were in $\mathcal{CH}_3 \backslash \mathcal{CH}_2$ we could simply check a member of each class and verify whether or not it is semi-Clifford. From the discussion above, we can restrict to $\Tilde{\pi} \in \mathcal{CH}_3 \backslash \mathcal{CH}_2$ (or else $\Tilde{\pi} \Tilde{d}$ would be trivially semi-Clifford), but at this point we have not put any limits on $\Tilde{d}$ except that it must be in $\mathcal{CH}$. The next few lemmas place restrictions on $\Tilde{d}$. 

\begin{Lem}\label{piXd}
    If $\pi \in \mathcal{CH}_3$ and $d$ is a diagonal matrix, $\pi d \in \mathcal{CH}_3 \iff \pi \vec{X} d \vec{X} d^{-1} \pi^{-1} \in \mathcal{CH}_2$ for all $X$ Pauli strings $\vec{X}$.
\end{Lem}
\begin{proof}
    By definition, $\pi d \in \mathcal{CH}_3 \iff \pi d \vec{Z} \vec{X} d^{-1} \pi^{-1} \in \mathcal{CH}_2$ for all Pauli strings $\vec{Z}, \vec{X}$. 
    
    Then,
    \begin{equation*}
    \begin{array}{ll}
        \quad \pi d \vec{Z} \vec{X} d^{-1} \pi^{-1} \\ 
        =\pi \vec{Z} (\pi^{-1} \pi) d \vec{X} d^{-1} \pi^{-1} \quad \quad \quad (\mbox{insert identity})\\
        =\pi \vec{Z} \pi^{-1} \pi (\vec{X} \vec{X}) d \vec{X} d^{-1} \pi^{-1} \quad \quad (\mbox{insert identity})\\
        =\pi \vec{Z} \pi^{-1} \pi \vec{X} (\pi^{-1} \pi) \vec{X} d \vec{X} d^{-1} \pi^{-1} \quad (\mbox{insert identity}) \\
        =(\pi \vec{Z} \pi^{-1}) (\pi \vec{X} \pi^{-1}) \pi \vec{X} d \vec{X} d^{-1} \pi^{-1}.
    \end{array}
    \end{equation*}
    The terms in parentheses on the final line are both Clifford since $\pi \in \mathcal{CH}_3$, and since Clifford gates form a group, we see that $\pi \vec{X} d \vec{X} d^{-1} \pi^{-1}$ is Clifford if and only if $\pi d \in \mathcal{CH}_3$.
\end{proof}

We start with the classification of diagonal gates in $\mathcal{CH}$ \cite{Cui2017}. While conjugation by permutations does not necessarily preserve the level of a diagonal gate, it does preserve the spectrum (entries of the diagonal gate). We find it useful to introduce the diagonal matrix group: $Diag^n_r$. This is the group of all diagonal $2^n \times 2^n$ matrices with entries in the $2^r$th roots-of-unity. From \cite{Cui2017} these diagonal matrices are all in the Clifford Hierarchy. We can relate diagonal matrices from $\mathcal{CH}_k$ and $Diag^n_r$ as the following lemma shows:

\begin{Lem}\label{D3/D2} A matrix in $Diag_3 \backslash Diag_2$ is in $\mathcal{D}_k\backslash \mathcal{D}_2$ for some $k\ge 3$.
\end{Lem}
\begin{proof}
 Additional properties of $\mathcal{D}_k$ and $Diag_k$ are discussed in Appendix \ref{sec:diag}.  

 The intuition for this Lemma is seen in Figure \ref{fig:diagGroups}. All entries in $Diag_3 \backslash Diag_2$ (third row of \ref{fig:diagGroups}(b)) are in at least $\mathcal{D}_k \backslash \mathcal{D}_{2}$ for some $k\ge 3$. 

 Recall that $Diag_2$ are diagonal matrices with entries from the $2^2 = 4$th roots of unity\footnote{By fixing the global phase, we can set the upper left entry of these matrices to 1.} $\mathcal{D}_2$ (the diagonal Clifford group) is generated by $\langle S_i, CZ_{ij}\rangle$. We can then see that $Diag_2 \supseteq \mathcal{D}_2$. Since $Diag_2$ contains all diagonal Clifford elements any element of $Diag_3\backslash Diag_2$ must therefore be non-Clifford which is the same as being a member of  $\mathcal{D}_k\backslash\mathcal{D}_2$ for some $k\ge 3$. 
\end{proof}

From Lemmas \ref{piXd} and \ref{D3/D2} we can exclude diagonal equivalence classes $\Tilde{d}$ if they contain a representative $d$ such that $\vec{X} d \vec{X} d^{-1} \in Diag_3 \backslash Diag_2$ for some $\vec{X}$, since this means that $\vec{X} d \vec{X} d^{-1}$ contains entries which are strictly 8th roots of unity ($\exp (i a \frac{2\pi}{8})$ for odd integer $a$). Since any diagonal matrix with such entries is in $\mathcal{CH}_k\backslash\mathcal{CH}_2$ for some $k\ge3$ and conjugation by a permutation preserves this property, we conclude that $\pi \vec{X} d \vec{X} d^{-1} \pi^{-1} \notin \mathcal{CH}_2$ in this case. 

Specifically, we see that there exists $\vec{X}$s such that $X \sqrt{T} X \sqrt{T^\dagger} \in Diag_3 \backslash Diag_2$ and 
\begin{equation*}
\begin{array}{cccc}
\Qcircuit @C=0.5em @R=1.3em {
& \ctrl{1} & \gate{X} & \ctrl{1}  & \qw \\
& \gate{T} &  \qw  &  \gate{T^\dagger} & \qw\\
} 
\end{array}
= 
\begin{array}{cccc}
\Qcircuit @C=0.5em @R=1.3em {
& \gate{X}        & \qw & \ctrl{1}  & \qw \\
& \gate{T} & \qw &  \gate{S^\dagger} & \qw\\
}  
\end{array}
\in Diag_3 \backslash Diag_2.
\end{equation*}

Combining these results we find that there are at most $2^{20}$ diagonal equivalence classes to consider, which we have organized in Fig.~\ref{fig:classes1}. 

\begin{figure}[!h]
\begin{equation*}
\begin{array}{ccc}
\Qcircuit @C=0.5em @R=1.3em {
& \qw & \ctrl{1} & \ctrl{1}  & \ctrl{1} & \ctrl{2} &\qw       & \ctrl{1} & \ctrl{2} & \ctrl{3} & \qw     & \qw & \qw & \gate{T} & \qw & \qw & \qw \\
& \ctrl{1} & \ctrl{1} & \ctrl{1}  & \ctrl{2} & \qw      & \ctrl{1} & \gate{S} & \qw      & \qw      & \ctrl{1} & \ctrl{2} & \qw & \qw & \gate{T} & \qw & \qw \\
& \ctrl{1} & \ctrl{1} & \gate{S}  & \qw      & \ctrl{1} & \ctrl{1} & \qw      & \gate{S} & \qw      & \gate{S} & \qw & \ctrl{1} & \qw & \qw & \gate{T} & \qw \\
& \targ & \gate{S} &  \qw      & \gate{S} & \gate{S} & \gate{S} & \qw      & \qw      & \gate{S} & \qw & \gate{S} & \gate{S} & \qw & \qw & \qw & \gate{T}\\
& & a_1 & a_2 & a_3 & a_4 & a_5 & b_1 & b_2 & b_3 & b_4 & b_5 & b_6 & c_1 & c_2 & c_3 & c_4 \\
}
\end{array}
\end{equation*}
\caption{Here $a_i \in \{0,1,2,3\}$ and $b_i, c_i \in \{0,1\}$ represent the exponent of the gate above. $b_i, c_i$ only require two values since the other powers are equivalent up to Clifford diagonal gates. We have included a permutation from the relevant affine equivalence class to illustrate how the equivalence classes interact.}
\label{fig:classes1}
\end{figure}

Certain permutations, $\pi$, can further restrict the number of diagonal equivalence classes we need to consider as the following lemma shows:

\begin{Lem}\label{noC3}
All diagonal gates, $d_{\Bar{T}}$, in $\mathcal{CH}_3$ which have no support on the targets of a permutation, $\pi$, can be ignored when determining whether $\pi d \in \mathcal{CH}_k$.
\end{Lem}
\begin{proof}
    Let $\pi$ be a product of $C^n(X)$ for $n\ge 2$ gates. Up to multiplication by Clifford permutations this is a fully general permutation. Let $T$ denote the wires (qubits) which have support on a target of $\pi$ and let $\bar{T}$ denote all other wires. For a generalized semi-Clifford gate, $\pi d=\pi d_T d_{\bar{T}}$, with a diagonal component trivial support on $T$, $d_{\Bar{T}}$, we can check $d_{\Bar{T}}$ membership in $\mathcal{CH}$ separately for all $X$ as the following identity shows:
    \begin{align*} & \quad \pi d_{T} d_{\Bar{T}} \vec{X} d_{\Bar{T}}^{-1} d_T^{-1} \pi^{-1} \\ &= \pi d_T (d_{\Bar{T}} \vec{X} d_{\Bar{T}}^{-1} \vec{X}) \vec{X}d_T^{-1}\pi^{-1} \\ &=(d_{\Bar{T}} \vec{X} d_{\Bar{T}}^{-1} \vec{X}) \pi d_T \vec{X} d_T^{-1} \pi^{-1}.
    \end{align*}
    Then, if $d_{\Bar{T}} \in \mathcal{CH}_3$, then $(d_{\Bar{T}} \vec{X} d_{\Bar{T}}^{-1} \vec{X}) \in \mathcal{CH}_2$ and for all $X$, $\pi d \in \mathcal{CH}_k$ iff $\pi d_T \in \mathcal{CH}_k$.
\end{proof}

Now, we apply Lemma~\ref{noC3} to the diagonal equivalence classes shown in Fig.~\ref{fig:classes1} for the case with a single target on the bottom most qubit. 

\begin{figure}[!h]
\begin{equation*}
\begin{array}{ccc}
\Qcircuit @C=0.5em @R=1.3em {
& \qw      & \ctrl{1} & \ctrl{1} & \ctrl{1} & \ctrl{2} &\qw        & \ctrl{3} & \qw & \qw & \qw \\
& \ctrl{1} & \ctrl{1} & \ctrl{1} & \ctrl{2} & \qw      & \ctrl{1} & \qw      & \ctrl{2} & \qw & \qw \\
& \ctrl{1} & \gate{S} & \ctrl{1} & \qw      & \ctrl{1} & \ctrl{1} & \qw      & \qw & \ctrl{1} & \qw \\
& \targ    & \qw      & \gate{S} & \gate{S} & \gate{S} & \gate{S} & \gate{S} & \gate{S} & \gate{S} & \gate{T}\\
& & a_1 & a_2 & a_3 & a_4 & a_5 & b_1 & b_2 & b_3 & c_1 \\
}
\end{array}
\end{equation*}
\caption{Here $a_i \in \{0,1,2,3\}$ and $b_i, c_i \in \{0,1\}$ represent the exponent of the gate above. $b_i, c_i$ only require two values since the other powers are equivalent up to Clifford diagonal gates. We have $4^5 2^4 = 16384$ classes to check. Note, we have included a permutation from the relevant affine equivalence class to illustrate how the equivalence classes interact.}
\label{fig:classes2}
\end{figure}


We can express the diagonal matrix as
\begin{multline*}
d=\text{diag}(1, \; c_1, \; 1, \; b_3 c_1, \;
    1, b_2 c_1, 1, a_5 b_2 b_3 c_1, 1,\; b_1 c_1,\; 1,\; a_4 b_1 b_3 c_1,\\
    1,\; a_3 b_1 b_2 c_1,\; a_1, \; a_1 a_2 a_3 a_4 a_5 b_1 b_2 b_3 c_1)
\end{multline*}
where $a_i \equiv e^{i\frac{a_i 2\pi}{4}}$, $b_i \equiv e^{i\frac{b_i 2\pi}{4}}$, and $c_i \equiv e^{i\frac{c_i 2\pi}{8}}$. For all choices of $a_i, b_i, c_i$ above (except for setting all variables to zero), $d$ is in $\mathcal{CH}_k\backslash\mathcal{CH}_2$. Almost all choices of $a_i, b_i, c_i$ when multiplied by the $CCX$ permutation gate (or any permutation gate in this affine equivalence class) are not semi-Clifford except for $a_5 = 2$ and all other variables set to zero. This case is semi-Clifford since $CCX\cdot CCZ = CCiY$ which is diagonalized by Clifford gates. Excluding this case we can now iterate over all equivalence classes in the algorithm below. 

For completeness, we list some additional ideas for cutting down on the total number of equivalence classes.  

\begin{Lem}
$\pi d \in \mathcal{CH}_k \iff \pi d^{-1} \in \mathcal{CH}_k$.    
\end{Lem}

This is a simple corollary of the theorem that complex conjugation preserves the level in the Clifford hierarchy. See \cite{SEInverses} for details.

This lemma implies that we do not have to check all values of $a_i$. While useful, it does not significantly reduce the total number of gates to check.

Additionally, since we know that $\mathcal{CH}_3^{n=3}$ is semi-Clifford, if all gates ($\pi$ and the diagonal gates) have support on only the same 3 qubits, we do not need to check this case. This fact can also be used to exclude the $a_5 = 2$ and all other variables set to zero case discussed above. In fact, any combination of $a_5, b_2, b_3, \mbox{ and } c_1$ with all other variables set to zero can be excluded.

\begin{algorithm}
  \caption{Check All the Gates }
  \label{euclid}
  \begin{algorithmic}
      \For{\texttt{$\pi \in \mathcal{CH}_3^{(n)}$}} \Comment{Loop over candidate permutations in $\mathcal{CH}_3^{(n)}$}
        \For{\texttt{$d \in \mathcal{CH}^{(n)}_k$}} \Comment{Loop over candidate diagonal matrices in $\mathcal{CH}^{(n)}_k$}
            \State $Q \gets True$
            \For{\texttt{$\vec{X} \in \text{Pauli X Strings}$}} \Comment{Loop over all $2^n-1$ non-identity $\vec{X}$}
        \State \texttt{Check if $\pi \vec{X} d \vec{X} d^{-1} \pi^{-1}$ is Clifford}
        \If{\texttt{Not Clifford}}
            \State $Q \gets False$
            \State \texttt{Break out of inner most For loop}
        \EndIf
        \EndFor
        \If{$Q$\texttt{ is True} }
            \If{$\pi d$  passes semi-Clifford check\footnotemark}
            \\ \hspace{0.85in}
                \Return{\texttt{non-semi-Clifford gate found!}}
            \EndIf
        \EndIf
        \EndFor
      \EndFor\\
      \Return{\texttt{All gates are semi-Clifford}}
  \end{algorithmic}
\end{algorithm}
\footnotetext{This semi-Clifford check is more computationally expensive and we use it sparingly.}
We implemented the algorithm above and find that all gates in $\mathcal{CH}_3^{n=4}$ are semi-Clifford.

We used the classification of 4-qubit affine equivalence classes of permutations in the Clifford hierarchy (see Section~\ref{Perms4Qs}) to ensure that our algorithm iterated over all relevant equivalence classes, and acquired our result via exhaustive search.

\section{Computational attack on $\mathcal{CH}_3^{n=5}$}\label{sec:atk35}

We can see from Table~\ref{table:numAE} that the number of AE classes for 5 qubits makes the strategy of checking each AE class for membership in $\mathcal{CH}_3\backslash\mathcal{CH}_2$ (as we did in the previous section) intractable. We can, however, leverage results from \cite{He2024} where they show that all permutations in $\mathcal{CH}_3$ (and by extension $\mathcal{CH}_3\backslash\mathcal{CH}_2$) can be written in staircase form up to left/right multiplication by Clifford permutations.

Staircase form is defined in \cite{He2024} as follows:

\begin{Def}[Staircase form]
     A product of pairwise distinct Toffoli gates is said to be in {\bf staircase form} if each Toffoli gate TOF$_{i,j,k}$ in the product has $i < j<k$ and the target qubits are in non-decreasing order in the order that the gates are applied (left to right in circuit diagrams). 
\end{Def}

Additionally, we can use affine equivalence to further reduce the number of cases we must check. After doing so, we find three AE classes in $\mathcal{CH}_3\backslash\mathcal{CH}_2$. The Theorems and Lemmas from the previous section can be applied to reduce the number of diagonal gates which must be checked. However, the diagonal gates will be slightly different for each permutation class. We illustrate this in the circuits below. 

\begin{figure}[!h]
\begin{equation*}
\begin{array}{ccc}
\Qcircuit @C=0.4em @R=0.9em {
& \ctrl{1} & \qw      & \ctrl{1} & \ctrl{1} & \ctrl{1} & \ctrl{2} & \qw      & \ctrl{1} & \ctrl{1} & \ctrl{2} & \ctrl{3} & \qw      & \qw      & \qw      & \ctrl{1} & \ctrl{1} & \ctrl{2} & \qw     &\ctrl{4}& \qw     & \qw      & \qw      & \qw \\ 
& \ctrl{3} & \qw      & \ctrl{1} & \ctrl{1} & \ctrl{2} & \qw      & \ctrl{1} & \ctrl{1} & \ctrl{3} & \qw      & \qw      & \ctrl{1} & \ctrl{2} & \qw      & \ctrl{1} & \ctrl{2} & \qw      & \ctrl{1}&\qw     & \ctrl{3}& \qw      & \qw      & \qw \\
& \qw      & \ctrl{1} & \ctrl{1} & \ctrl{2} & \qw      & \ctrl{1} & \ctrl{1} & \ctrl{1} & \qw      & \ctrl{2} & \qw      & \ctrl{2} & \qw      & \ctrl{1} & \gate{S} & \qw      & \ctrl{1} & \ctrl{1}&\qw     & \qw      & \ctrl{2}& \qw      & \qw \\
& \qw      & \ctrl{1} & \ctrl{1} & \qw      & \ctrl{1} & \ctrl{1} & \ctrl{1} & \gate{S} & \qw      & \qw      & \ctrl{1} & \qw      & \ctrl{1} & \ctrl{1} & \qw      & \gate{S} & \gate{S} & \gate{S}&\qw     & \qw      & \qw     & \ctrl{1} & \qw \\
& \targ    & \targ    & \gate{S} & \gate{S} & \gate{S} & \gate{S} & \gate{S} & \qw      & \gate{S} & \gate{S} & \gate{S} & \gate{S} & \gate{S} & \gate{S} & \qw      & \qw      & \qw      & \qw     &\gate{S}& \gate{S} & \gate{S}& \gate{S} & \gate{T}\\
&          &          & a_1      & a_2      & a_3      & a_4      & a_5      & a_6'     & a_7      & a_8      & a_9      & a_{10}   & a_{11}   & a_{12}   & a_{13}'  & a_{14}'  & a_{15}'  & a_{16}' & b_{1} & b_{2} & b_{3} & b_{4} & b_{5} }
\end{array}
\end{equation*}
\caption{Here $a_i \in \{0,1,2,3\}$ and $b_i \in \{0,1\}$ represent the exponent of the gate above. $b_i$ only require two values since the other powers are equivalent up to Clifford diagonal gates. Primed variables commute with permutation and any case consisting only of only non-zero entries on primed variables can be skipped since it is already semi-Clifford. We have included a permutation from the relevant affine equivalence class to illustrate how the equivalence classes interact.}
\label{fig:classes3}
\end{figure}

\begin{figure}[!h]
\begin{equation*}
\begin{array}{ccc}
\Qcircuit @C=0.3em @R=0.9em {
& \ctrl{1} & \ctrl{3} & \ctrl{1} & \ctrl{1} & \ctrl{1} & \ctrl{2} & \qw      & \ctrl{1} & \ctrl{1} & \ctrl{2} & \ctrl{3} & \qw      & \qw      & \qw      & \ctrl{1} & \ctrl{1} & \ctrl{2} & \qw     &\ctrl{4}& \qw     & \qw      & \ctrl{2} & \qw & \qw & \qw \\ 
& \ctrl{1} & \qw      & \ctrl{1} & \ctrl{1} & \ctrl{2} & \qw      & \ctrl{1} & \ctrl{1} & \ctrl{3} & \qw      & \qw      & \ctrl{1} & \ctrl{2} & \qw      & \ctrl{1} & \ctrl{2} & \qw      & \ctrl{1}&\qw     & \ctrl{3}& \qw      & \qw      & \ctrl{1} & \qw & \qw \\
& \targ    & \qw & \ctrl{1} & \ctrl{2} & \qw      & \ctrl{1} & \ctrl{1} & \ctrl{1} & \qw      & \ctrl{2} & \qw      & \ctrl{2} & \qw      & \ctrl{1} & \gate{S} & \qw      & \ctrl{1} & \ctrl{1}&\qw     & \qw      & \ctrl{2} & \gate{S} & \gate{S} & \gate{S}  & \gate{T} \\
& \qw      & \ctrl{1} & \ctrl{1} & \qw      & \ctrl{1} & \ctrl{1} & \ctrl{1} & \gate{S} & \qw      & \qw      & \ctrl{1} & \qw      & \ctrl{1} & \ctrl{1} & \qw      & \gate{S} & \gate{S} & \gate{S}&\qw     & \qw      & \qw     & \ctrl{1} & \qw     & \ctrl{-1} & \qw \\
& \qw    & \targ    & \gate{S} & \gate{S} & \gate{S} & \gate{S} & \gate{S} & \qw      & \gate{S} & \gate{S} & \gate{S} & \gate{S} & \gate{S} & \gate{S} & \qw      & \qw      & \qw      & \qw     &\gate{S}& \gate{S} & \gate{S}& \gate{S} & \qw      & \qw      & \gate{T} \\
&          &          & a_1      & a_2      & a_3      & a_4      & a_5      & a_6'     & a_7      & a_8      & a_9      & a_{10}   & a_{11}   & a_{12}   & a_{13}'  & a_{14}'  & a_{15}'  & a_{16}' & b_{1} & b_{2} & b_{3} & b_{4,5} & b_6 & b_7 & b_{8,9} }
\end{array}
\end{equation*}
\caption{Here $a_i \in \{0,1,2,3\}$ and $b_i \in \{0,1\}$ represent the exponent of the gate above. $b_i$ only require two values since the other powers are equivalent up to Clifford diagonal gates. Primed variables commute with permutation and any case consisting only of only non-zero entries on primed variables can be skipped since it is already semi-Clifford. We have included a permutation from the relevant affine equivalence class to illustrate how the equivalence classes interact.}
\label{fig:classes4}
\end{figure}

\begin{figure}[!h]
\begin{equation*}
\begin{array}{ccc}
\Qcircuit @C=0.4em @R=0.9em {
& \qw      & \ctrl{1} & \ctrl{1} & \ctrl{1} & \ctrl{2} & \qw      & \ctrl{1} & \ctrl{1} & \ctrl{2} & \ctrl{3} & \qw      & \qw      & \qw      & \ctrl{1} & \ctrl{1} & \ctrl{2} & \qw     &\ctrl{4}& \qw     & \qw      & \qw      & \qw \\ 
& \qw      & \ctrl{1} & \ctrl{1} & \ctrl{2} & \qw      & \ctrl{1} & \ctrl{1} & \ctrl{3} & \qw      & \qw      & \ctrl{1} & \ctrl{2} & \qw      & \ctrl{1} & \ctrl{2} & \qw      & \ctrl{1}&\qw     & \ctrl{3}& \qw      & \qw      & \qw \\
& \ctrl{1} & \ctrl{1} & \ctrl{2} & \qw      & \ctrl{1} & \ctrl{1} & \ctrl{1} & \qw      & \ctrl{2} & \qw      & \ctrl{2} & \qw      & \ctrl{1} & \gate{S} & \qw      & \ctrl{1} & \ctrl{1}&\qw     & \qw      & \ctrl{2}& \qw      & \qw \\
& \ctrl{1} & \ctrl{1} & \qw      & \ctrl{1} & \ctrl{1} & \ctrl{1} & \gate{S} & \qw      & \qw      & \ctrl{1} & \qw      & \ctrl{1} & \ctrl{1} & \qw      & \gate{S} & \gate{S} & \gate{S}&\qw     & \qw      & \qw     & \ctrl{1} & \qw \\
& \targ    & \gate{S} & \gate{S} & \gate{S} & \gate{S} & \gate{S} & \qw      & \gate{S} & \gate{S} & \gate{S} & \gate{S} & \gate{S} & \gate{S} & \qw      & \qw      & \qw      & \qw     &\gate{S}& \gate{S} & \gate{S}& \gate{S} & \gate{T}\\
&          & a_1      & a_2      & a_3      & a_4      & a_5      & a_6'     & a_7      & a_8      & a_9      & a_{10}   & a_{11}   & a_{12}   & a_{13}'  & a_{14}'  & a_{15}'  & a_{16}' & b_{1} & b_{2} & b_{3} & b_{4} & b_{5} }
\end{array}
\end{equation*}
\caption{Here $a_i \in \{0,1,2,3\}$ and $b_i \in \{0,1\}$ represent the exponent of the gate above. $b_i$ only require two values since the other powers are equivalent up to Clifford diagonal gates. Primed variables commute with permutation and any case consisting only of only non-zero entries on primed variables can be skipped since it is already semi-Clifford. We have included a permutation from the relevant affine equivalence class to illustrate how the equivalence classes interact.}
\label{fig:classes5}
\end{figure}


Similarly to section \ref{sec:atk34}, we use Algorithm \ref{euclid} to iterate through all candidate permutations in $\mathcal{CH}_3^{n=5}$. This is made tractable by leveraging the results from \cite{He2024}. Checking if all $n=6$ equivalence classes are strictly semi-Clifford likely requires an improved theoretical understanding to eliminate classes as well as abundant computational resources since the number of diagonal equivalence classes to consider grows rapidly with $n$.

\newpage

\section{Cycle Structures}
\begin{Def}
A \textbf{cycle structure} for a permutation, $P$, is a matrix comprised of all states (expressed as binary column vectors) that $P$ does not fix. 
\end{Def}
When a cycle structure has only one non-trivial cycle of length $k$, the cycle structure has $k$ columns and we refer to it as a $k$ cycle structure. Likewise, when a cycle structure contains non-trivial disjoint cycles of lengths: $(k_1, k_2,...,k_m)$, we refer to it as having \textbf{cycle type} $(k_1, k_2,...,k_m)$, and \textbf{length} $\sum_{i=1}^m k_i$, which is equal to the number of columns in the binary matrix. 

A cycle structure is unique up to cyclic permutation of states within each contained cycle and up to changing the order of different disjoint cycles. A {\em canonical cycle notation} can be used to uniquely describe these cycles, which is similar to that used in finite group theory, where:
\begin{itemize}
    \item The largest (interpreting each binary vector as an integer) element of each disjoint cycle is listed first.
    \item Disjoint cycles are sorted in descending order by their first element.
    \item Fixed points, or $1$-cycles, are omitted.\footnote{Some canonical cycle notations in finite group theory require you write out all one-cycles, however this would be cumbersome when you consider the size of the state space fixed by our permutations of interest.}
\end{itemize}

For example, the cycle structure written in matrix form as 
$$\left[\begin{matrix}
    1 & 0 & 1 \\
    0 & 1 & 1 \\
    0 & 0 & 0 \\
    0 & 0 & 0 
\end{matrix} \ \right|  \left.  
\begin{matrix}
    0 & 1 \\
    0 & 1 \\
    0 & 1 \\
    1 & 0 
\end{matrix} \ \right|  \left.  
\begin{matrix}
    1 & 1 \\
    1 & 0 \\
    1 & 1 \\
    1 & 0 
\end{matrix} \right],$$
has a cycle type of $(3,2,2)$, and can be represented in canonical cycle notation as

$$ (15,5)(8,7)(3,1,2). $$

Each column of any given cycle structure must be distinct. Columns need not be linearly independent. From this property it is easy to show that a $k$-cycle's cycle structure is at least rank $\lc\log_2(k)\rc$.

We use cycle structures, particularly in their matrix forms, to represent permutations because we are interested in affine equivalence of permutations, and conjugation by affine permutations has a particularly simple action on cycle structures in this form.

\begin{figure}
\begin{equation*}
\begin{bmatrix}
1 & 0 & 0 & 0 & 0 & 0 & 0 & 0 \\
0 & 1 & 0 & 0 & 0 & 0 & 0 & 0 \\
0 & 0 & 1 & 0 & 0 & 0 & 0 & 0 \\
0 & 0 & 0 & 1 & 0 & 0 & 0 & 0 \\
0 & 0 & 0 & 0 & 1 & 0 & 0 & 0 \\
0 & 0 & 0 & 0 & 0 & 1 & 0 & 0 \\
0 & 0 & 0 & 0 & 0 & 0 & 0 & 1 \\
0 & 0 & 0 & 0 & 0 & 0 & 1 & 0 \\
\end{bmatrix} 
\Leftrightarrow 
\begin{array}{ccc}
\Qcircuit @C=0.5em @R=1.3em {
& \ctrl{1} & \qw \\
& \ctrl{1}& \qw \\
& \targ& \qw \\
}
\end{array} 
\Leftrightarrow 
\begin{bmatrix}
1 & 1 \\
1 & 1 \\
0 & 1 \\
\end{bmatrix}
\end{equation*}
\caption{Comparison of three different representations of the Toffoli gate. On the left is the canonical unitary representation, in the middle is the quantum gate representation, and on the right is our cycle structure representation.}
\end{figure}

\begin{theorem}\label{kcycles}
All length-$k$ cycles structures on $m > k$ qubits are affine equivalent to a controlled permutation. Furthermore, there are only a constant (dependent on $k$, independent of $m$) number of affine equivalence classes of permutations corresponding to length-$k$ cycle structures. All others can be obtained by adding controls to permutations in this set.  
\end{theorem}

In order to prove this, we introduce some useful lemmas:
\begin{Lem}\label{rreLem}
A binary matrix can be put into reduced row echelon form using elementary row operations. The rank of a matrix in reduced row echelon form is the number of pivots (leading ones in each row). 
\end{Lem}

\begin{Lem}\label{const}
A constant row (all 0's or all 1's) of a cycle structure matrix corresponds to a controlled permutation with the constant row acting as the control (or anticontrol) and the remaining row(s) corresponding to some permutation.  
\end{Lem}
\begin{proof}
    Let there be an $m$-cycle of interest with a constant row. Assume without loss of generality that the constant row is the first one, and that it contains all ones. Thus the cycle structure in matrix form looks like this:
    $$
    \left[
    \begin{array}{ccccc}
    1     & 1        &  ...  &  1    &  1 \\ \cline{1-5}
    \bordl      &          &       &   & \bordr    \\ 
    \bordl &          &  \sigma     &    & \bordr  \\ 
    \bordl      &          &       &  & \bordr     \\ \cline{1-5}
    \ & \ & \ & \  & \\
  \end{array}\right]
  $$
    Where $\sigma$ is the cycle structure of the $m$-cycle on the non-fixed qubits. We can see that the quantum circuit that implements this cycle structure must be of the form
    $$
    G = \ket{0}\bra{0} \otimes I + \ket{1}\bra{1} \otimes \sigma
    $$
    since if the first qubit is in state $\ket{1}$, $\sigma$ is applied to permute the remaining qubits, and if the first qubit is in $\ket{0}$, all states must be fixed. This is the definition of a controlled gate on the first qubit, and we thus have that $G$ is a controlled-$\sigma$ gate.
\end{proof}
Now we prove Theorem~\ref{kcycles}.

\begin{proof}
Affine equivalence allows for CNOTs between rows of the cycle structure. Hence we have SWAPs between rows and row addition (modulo 2) which provide all the elementary row operations for a binary matrix. In addition to elementary row operations, we have an affine operation of adding $1$ (addition modulo $2$) to all entries in a row and we can permute columns cyclically. 

From Lemma~\ref{rreLem} we can put any length-$k$ cycle structure into reduced row echelon form. Additionally, we can always permute columns and/or SWAP rows to put a 1 in the upper left corner of the matrix. The placement of the other pivots within the reduced row echelon form is not fixed, but is easily bounded. The non-pivot rows can be any binary vector with 0's below the nearest pivot to the left. 

Let $\mathcal{A}_k$ contain one representative of each affine equivalence class of $r \times k$ binary matrices with $r\le k$. Note that matrices in $\mathcal{A}_k$ have rank $\leq{k}$. 

Now take any $(k+1) \times k$ binary matrix corresponding to a length-$k$ cycle structure on $k+1$ qubits. Using Lemma~\ref{rreLem} we can bring this into reduced row echelon form which will have at least one row of all 0's. From Lemma~\ref{const} this a controlled permutation. Furthermore, this controlled permutation (the $r\le k$ non-constant rows of the matrix) is affine equivalent to a permutation in $\mathcal{A}_k$. By repeating this procedure we can see that one of the following must be true for any $m \times k$ binary matrix with $m > k$: By adding an additional row, either the rank will increase (but remain strictly $\leq k$) to make a controlled permutation which is affine equivalent to a higher-ranked element of $\mathcal{A}_k$ or another control will be added making a multi-controlled permutation which is itself affine equivalent to an element in $\mathcal{A}_k$.
\end{proof}

\begin{theorem}\label{ctrlCH}
    From Anderson \& Weippert\cite{Anderson2024a}: A controlled-$U$ gate is in $\mathcal{CH}$ only if $U\in \mathcal{CH}$ and $U^{2^k} = I$ for some nonnegative integer $k$.
\end{theorem}

\begin{corollary}\label{CorCS}
From theorems \ref{kcycles} and \ref{ctrlCH} above it follows that any length-$k$ cycle structure on $m > k$ qubits is in $\mathcal{CH}$ only if it has cycle type $(2^{l_0},2^{l_1}, \ \dots \ ,2^{l_N})$, where each integer $l_i \ge 1$.
\end{corollary}

We can actually improve Corollary \ref{CorCS} by considering the case of a length-$k$ cycle structure on $m=k$ qubits. If this cycle structure has rank less than $m$ then it must be affine equivalent to a controlled permutation. Therefore, we only need to consider the full-rank case. From the discussion above this cycle structure is affine equivalent to the identity matrix which itself does not have a constant row and does not in its current form correspond to a controlled permutation. However, by conjugating the bottom ($m$th) qubit with $X$ and conjugating by a series of $CNOT$s each with a control on one of the other $m-1$ qubits and a target on the $m$th qubit we can see that this permutation is affine equivalent to a controlled permutation. Our slightly improved corollary now reads as follows:

\begin{corollary}\label{CorCS2}
From theorems \ref{kcycles} and \ref{ctrlCH} above it follows that any length-$k$ cycle structure on $m > k-1$ qubits is in $\mathcal{CH}$ only if it has cycle type $(2^{l_0},2^{l_1}, \ \dots \ ,2^{l_N})$, where each integer $l_i \ge 1$.
\end{corollary}

\begin{Lem}\label{lem:CS3}
The number of affine equivalence classes for a length-$k$ cycle structure is constant for $m>k-1$ qubits.  
\end{Lem}
From theorem \ref{kcycles} and corollary \ref{CorCS2} above it follows that any $(m>k-1)$-qubit permutation corresponding to a length-$k$ cycle structure is affine equivalent to a controlled permutation. Then, all affine equivalence classes for $m>k-1$ qubits are obtained by adding controls to existing permutations.

\section{$N$-qubit Classification of Cycle Structures}
From the last section we know that all length-$k$ cycle structures on $m \ge k-1$ qubits are controlled permutations. Furthermore, we know that the number of affine equivalence classes for a length-$k$ cycle structure on  $m \ge k-1$ qubits is constant. This provides a strategy for finding all affine equivalence classes for a fixed cycle type. In this section, we will use the affine equivalence testing algorithm of Canniere~\cite{Canniere2007} to check all permutations of a fixed cycle type for up to $m = k-1$ qubits. For $m > k-1$ we can then use Theorem~\ref{ctrlCH} to exclude affine equivalence classes that are not in $\mathcal{CH}$ and to exclude cycle structures that have order which is not a power of two. When cycle structures remain, we must prove whether or not adding a control always keeps this cycle structure in the Clifford Hierarchy.     

Using this strategy, we classified all cycle structures with length up to $6$. The table below summarizes our results as a function of number of qubits, $n$, and cycle type $c$. Note that non-trivial entries of the table start at $3$ qubits, since all permutations are affine equivalent on $2$ or fewer qubits.

\begin{table}[h!]
\centering
\begin{tabular}{l|c|c|c|c|c|c|c|c|c|c|c|}
\diagbox{$n$}{$c$} & \text{Id.} & $(2)$  & $(3)$  & $(4)$ & $(2,2)$ & $(5)$ & $(2,3)$ & $(6)$ & $(4,2)$ & $(3,3)$ & $(2,2,2)$ \\
\hline
$n=1$ & $\textcolor{magenta}{1}/1$ & $\textcolor{magenta}{1}/1$ & $0$ & $0$ & $0$ & $0$ & $0$ & $0$ & $0$ & $0$ & $0$ \\
\hline
$n=2$ & $\textcolor{magenta}{1}/1$ & $\textcolor{magenta}{1}/1$ & $\textcolor{magenta}{1}/1$  & $\textcolor{magenta}{1}/1$ & $\textcolor{magenta}{1}/1$ & $0$ & $0$ & $0$ & $0$ & $0$ & $0$ \\
\hline
$n=3$ & $\textcolor{magenta}{1}/1$ & $\textcolor{magenta}{1}/1$ & $\textcolor{magenta}{0}/1$ & $\textcolor{magenta}{1}/2$ & $\textcolor{magenta}{1}/2$ & $\textcolor{magenta}{0}/1$ &
$\textcolor{magenta}{1}/2$ &
$\textcolor{magenta}{1}/2$ & 
$\textcolor{magenta}{1}/2$ &
$\textcolor{magenta}{1}/2$ &
$\textcolor{magenta}{1}/2$\\
\hline
$n=4$ & $\textcolor{magenta}{1}/1$ & $\textcolor{magenta}{1}/1$ & $\textcolor{magenta}{0}/1$ & $\textcolor{magenta}{1}/2$ & $\textcolor{magenta}{1}/2$ & $\textcolor{magenta}{0}/2$ &
$\textcolor{magenta}{0}/3$ &
$\textcolor{magenta}{0}/9$ &
$\textcolor{magenta}{1}/9$ &
$\textcolor{magenta}{0}/6$ &
$\textcolor{magenta}{2}/6$\\
\hline
$n\geq 5$ & $\textcolor{magenta}{1}/1$ & $\textcolor{magenta}{1}/1$ & $\textcolor{magenta}{0}/1$ & $\textcolor{magenta}{1}/2$ & $\textcolor{magenta}{1}/2$ & $\textcolor{magenta}{0}/2$ &
$\textcolor{magenta}{0}/3$ &
$\textcolor{magenta}{0}/10$ & $\textcolor{magenta}{1}/10$ & $\textcolor{magenta}{0}/7$& $\textcolor{magenta}{2}/7$\\ 
\hline
\end{tabular}
\caption{In black (right): the number of Affine Equivalence classes of cycle structures for varying numbers of qubits and cycle types. In \textcolor{magenta}{magenta (left)}: the number of those cycle structures which are in $\mathcal{CH}$. Cycle structures are considered with length up to $6$ to reduce computational requirements.}\label{table:Cycles} 
\end{table}

We used a computer to find the number of affine equivalence classes and also to test membership in $\mathcal{CH}$ for a representative of each class. 

We obtain nonzero \textcolor{magenta}{magenta} entries in the ``$n\ge 5$'' row by proving that adding any number of controls to elements in the affine equivalence class preserves membership in $\mathcal{CH}$. When the equivalence class is composed of semi-Clifford elements\footnote{The elements in an affine equivalence class are either all semi-Clifford or none of them are semi-Clifford}, we can appeal to the classification of diagonal gates in $\mathcal{CH}$ \cite{Cui2017}. We developed a technique to prove that certain non-semi-Clifford permutations are in $\mathcal{CH}$ (see \ref{App:1MM}), but a general method is still lacking.      

\subsection{Families of Permutations in $\mathcal{CH}$}
In this section, we provide circuit diagrams and cycle structure matrices for each family of permutations in $\mathcal{CH}$. 

\begin{Def}
We define a \textbf{circuit family} by a \textbf{generating unitary} and all members of the family are obtained by adding controls to the generating unitary.
\end{Def}

Here we are interested in families of permutation gates which have the property that all members of the family (on a finite number of qubits) are in the Clifford Hierarchy. We list the families of permutation gates with this property that we have found below.

\begin{figure}[!htb]
\begin{equation*}
\begin{array}{ccc}
\Qcircuit @C=0.5em @R=1.3em {
& & \\
& \ctrl{-1}\qwx[1] & \qw \\
& \ctrl{1} & \qw \\
& \targ & \qw }
\end{array} 
\Leftrightarrow 
\begin{bmatrix}
1 & 1 \\
1 & 1 \\
0 & 1 \\
\end{bmatrix}
\end{equation*}
\caption{The Toffoli gate is a 3-qubit gate in $\mathcal{CH}$ and has cycle type: $(2)$. By adding controls to this gate, conjugating by CNOT networks (linear permutations), and multiplying by Pauli $X$ strings (affine permutations) we can construct all cycle type $(2)$ permutations for any $n$. All such transpositions are in $\mathcal{CH}$.}
\end{figure}
\begin{figure}[!htb]

\begin{equation*}
\begin{array}{ccc}
\Qcircuit @C=0.5em @R=1.3em {
 & & \\
 & \qw & \qw \\
& \ctrl{-2}\qwx[1] & \qw \\
& \ctrl{1} & \qw\\
& \targ & \qw\\
}
\end{array}
\Leftrightarrow 
\left [ 
\begin{array}{cc|cc}
1 & 1 & 0 & 0\\
1 & 1 & 1 & 1\\
1 & 1 & 1 & 1\\
0 & 1 & 0 & 1\\
\end{array}
\right ]
\end{equation*}
\caption{The 3-qubit Toffoli gate on 4 qubits is in $\mathcal{CH}$ and has cycle type: $(2,2)$. We can add controls and multiply by Clifford permutations to construct a family of cycle type $(2,2)$ permutations in $\mathcal{CH}$ for any $n$. The top-most bare wire indicates that all members of this family must contain exactly one bare wire. However, in contrast to the cycle type $(2)$ case there is an another affine equivalence class of cycle type $(2,2)$ permutations which is not in $\mathcal{CH}$.}
\end{figure}
 
\begin{figure}[!htb]
\begin{equation*}
\begin{array}{ccc}
\Qcircuit @C=0.5em @R=1.3em {
& & & \\
& \ctrl{-1}\qwx[1] & \ctrl{-1}\qwx[1]  & \qw\\
& \ctrl{1} &  \ctrl{1} & \qw\\
& \targ & \ctrl{1} & \qw\\
& \qw & \targ & \qw\\
}
\end{array}
\Leftrightarrow 
\left [ 
\begin{array}{cccc}
1 & 1 & 1 & 1\\
1 & 1 & 1 & 1\\
0 & 1 & 0 & 1\\
0 & 0 & 1 & 1\\
\end{array}
\right ]
\end{equation*}
\caption{This 4-qubit gate has cycle type: $(4)$. By adding controls to the gate we can construct a family of permutations in $\mathcal{CH}$. This gate family is not semi-Clifford, but since all members of this family have 1-wire mismatch (see \ref{App:1MM} for definition) we can use Corollary.~\ref{Corr:Fam} to show that this gate family is in $\mathcal{CH}$.}
\end{figure}

\begin{figure}[!htb]
\begin{equation*}
\begin{array}{ccc}
\Qcircuit @C=0.5em @R=1.3em {
& & & & & &\\
& \qw      & \qw      & \ctrl{-1}\qwx[1] & \ctrl{-1}\qwx[1]  & \qw      & \qw \\
& \gate{X} & \qw      & \ctrl{1}          & \targ              & \qw      & \gate{X}\\
& \qw      & \ctrl{1} & \targ             & \qw                & \ctrl{1} & \qw\\
& \gate{X} & \targ    & \ctrlo{-1}        & \ctrl{-2}          & \targ    & \gate{X}\\
}
\end{array}
\Leftrightarrow 
\left [ 
\begin{array}{cc|cc|cc}
1 & 1 & 1 & 1 & 1 & 1\\
0 & 1 & 0 & 0 & 0 & 1\\
0 & 0 & 1 & 0 & 1 & 1\\
0 & 0 & 0 & 1 & 1 & 1\\
\end{array}
\right ]
\end{equation*}
\caption{This 4-qubit gate has cycle type: $(2,2,2)$. By adding controls to the gate we can construct a family of permutations in $\mathcal{CH}$. This gate family is not semi-Clifford, but we have drawn it to show that Corollary.~\ref{Corr:Fam} can be used to show that this gate family is in $\mathcal{CH}$.}
\end{figure}

\begin{figure}[!htb]
\begin{equation*}
\begin{array}{ccc}
\Qcircuit @C=0.5em @R=1.3em {
& & & & &\\
& \qw & \ctrlo{-1}\qwx[1] & \ctrlo{-1}\qwx[1] & \ctrl{-1}\qwx[1]  & \qw\\
& \qw & \ctrlo{1} & \ctrl{1} & \ctrlo{1} & \qw\\
& \targ & \ctrl{1} & \ctrlo{1} & \ctrlo{1} & \targ\\
& \ctrl{-1} & \targ & \targ & \targ & \ctrl{-1} \\
}
\end{array}
\Leftrightarrow 
\left [ 
\begin{array}{cc|cc|cc}
0 & 0 & 1 & 1 & 0 & 0\\
1 & 1 & 0 & 0 & 0 & 0\\
0 & 1 & 0 & 1 & 0 & 1\\
0 & 1 & 0 & 1 & 1 & 0\\
\end{array}
\right ]
\end{equation*}
\caption{This 4-qubit gate has cycle type: $(2,2,2)$ and is not affine equivalent to the permutation above. By adding controls to the gate we can construct a family of permutations in $\mathcal{CH}$. From the figure we can see that this gate family is semi-Clifford.}
\end{figure}

\begin{figure}[!htb]
\begin{equation*}
\begin{array}{ccc}
\Qcircuit @C=0.5em @R=1.3em {
& & \\
& \ctrlo{2}\qwx[-1] & \qw & \qw \\
& \qw & \targ & \qw \\
& \targ & \ctrlo{-3} & \qw\\
}
\end{array}
\Leftrightarrow 
\left [ 
\begin{array}{cccc|cc}
0 & 0 & 0 & 0 & 1 & 1\\
0 & 0 & 1 & 1 & 0 & 1\\
0 & 1 & 0 & 1 & 0 & 0\\
\end{array}
\right ]
\end{equation*}
\caption{This 3-qubit gate has cycle type: $(4,2)$. By adding controls to the gate we can construct a family of permutations in $\mathcal{CH}$. This gate family is not semi-Clifford, but we have drawn it to show that Corollary.~\ref{Corr:Fam} can be used to show that this gate family is in $\mathcal{CH}$.}
\end{figure}
\newpage

\section{Conclusion}
In this paper, we introduced tools from the study of affine equivalence to elucidate the structure of the Clifford hierarchy. We found all affine equivalence classes of 4-qubit permutations, showed that all 4- and 5-qubit gates in the third level of $\mathcal{CH}$ are semi-Clifford, found all permutations in $\mathcal{CH}$ with small cycle lengths, and proved a few theorems about cycle structures. Additionally, we showed that all 1-wire-mismatch permutations are in $\mathcal{CH}$. While we do not yet have a full classification of permutations in $\mathcal{CH}$, our work combined with~\cite{He2024} provides a classification in some restricted settings. Briefly, these settings are: all permutations with $n\le 4$,  all permutations with $k\le 3$ (He et al.~\cite{He2024}), and all permutations with cycle structures up to length 6. 

To extend the ideas discussed here, new theoretical insights ares likely necessary. One strategy is to look for additional properties that all permutations (or AE classes) in $\mathcal{CH}$ must satisfy. This would reduce the number of permutations (or AE classes) that need to be checked for $n\ge5$ perhaps making additional cases tractable. Another idea is to extend the staircase normal form introduced in~\cite{He2024} to $k=4$ permutations thereby revealing the structure of permutations in the fourth level of $\mathcal{CH}$. Further studies of cycle structures would strengthen (or disprove by counterexample) the necessity of the controlled-unitary conditions shown by Anderson and Weippert~\cite{Anderson2024a}. Finally, a tight lower bound on the level in the Clifford Hierarchy of a product of a permutation and a diagonal gate, $\pi d$, may make it possible to computationally check if $\mathcal{CH}_3^{n=6}$ is semi-Clifford by eliminating many of the diagonal gates which currently need to be considered. This list is in no way exhaustive and we look forward to seeing what new discoveries can be made by intrepid researchers.   

\section{Acknowledgments}
AI Declaration: No AI was used in the writing or research of the initial submission of this manuscript. Since then, we have used AI (ChatGPT 5+) to check our algorithms and reproduce many of our cycle structure results. 

We thank the referees for pointing out many places where our presentation was unclear. The paper is significantly better as a result. JTA thanks Matthew Weippert for helpful discussions about cycle structures.  

\bibliographystyle{plainurl}
\bibliography{references}

@article{Draper2009,
   author = {Thomas Gordon Draper},
   city = {USA},
   isbn = {9781109380774},
   note = {AAI3372837},
   publisher = {University of Maryland at College Park},
   title = {Nonlinear complexity of boolean permutations},
   year = {2009},
}

@article{Canniere2007,
   author = {Christophe de Canniere},
   institution = {KATHOLIEKE UNIVERSITEIT LEUVEN},
   title = {Analysis and design of symmetric encryption algorithms},
   url = {https://www.esat.kuleuven.be/cosic/publications/thesis-139.pdf},
   year = {2007},
}

@article{Anderson2024,
   author = {Jonas T. Anderson},
   doi = {10.22331/q-2024-06-13-1370},
   issn = {2521-327X},
   journal = {Quantum},
   month = {6},
   pages = {1370},
   title = {On Groups in the Qubit Clifford Hierarchy},
   volume = {8},
   year = {2024},
}

@misc{SEInverses,
title = {Inverses and the Clifford Hierarchy},
author = {Quantum Computing StackExchange},
url = {https://quantumcomputing.stackexchange.com/questions/29348/inverses-and-the-clifford-hierarchy},
}

@article{Pllaha2020,
   title={Un-Weyl-ing the Clifford Hierarchy},
   volume={4},
   ISSN={2521-327X},
   url={http://dx.doi.org/10.22331/q-2020-12-11-370},
   DOI={10.22331/q-2020-12-11-370},
   journal={Quantum},
   publisher={Verein zur Forderung des Open Access Publizierens in den Quantenwissenschaften},
   author={Pllaha, Tefjol and Rengaswamy, Narayanan and Tirkkonen, Olav and Calderbank, Robert},
   year={2020},
   month=dec, pages={370} }

@article{Rengaswamy2019,
   title={Unifying the Clifford hierarchy via symmetric matrices over rings},
   volume={100},
   ISSN={2469-9934},
   url={http://dx.doi.org/10.1103/PhysRevA.100.022304},
   DOI={10.1103/physreva.100.022304},
   number={2},
   journal={Physical Review A},
   publisher={American Physical Society (APS)},
   author={Rengaswamy, Narayanan and Calderbank, Robert and Pfister, Henry D.},
   year={2019},
   month=aug }

@article{Bengtsson2014,
   title={Order 3 symmetry in the Clifford hierarchy},
   volume={47},
   ISSN={1751-8121},
   url={http://dx.doi.org/10.1088/1751-8113/47/45/455302},
   DOI={10.1088/1751-8113/47/45/455302},
   number={45},
   journal={Journal of Physics A: Mathematical and Theoretical},
   publisher={IOP Publishing},
   author={Bengtsson, Ingemar and Blanchfield, Kate and Campbell, Earl and Howard, Mark},
   year={2014},
   month=oct, pages={455302} }

@article{Anderson2016,
   author = {Jonas T. Anderson and Tomas Jochym-O'Connor},
   doi = {10.26421/QIC16.9-10-3},
   issn = {15337146},
   issue = {9&10},
   journal = {Quantum Information and Computation},
   month = {7},
   pages = {771-802},
   title = {Classification of transversal gates in qubit stabilizer codes},
   volume = {16},
   url = {http://www.rintonpress.com/journals/doi/QIC16.9-10-3.html},
   year = {2016},
}

@article{Anderson2024a,
      title={Controlled Gates in the Clifford Hierarchy}, 
      author={Jonas T. Anderson and Matthew Weippert},
      year={2025},
      eprint={2410.04711},
      archivePrefix={arXiv},
      primaryClass={quant-ph},
      url={https://arxiv.org/abs/2410.04711}, 
}

@misc{He2024,
      title={Characterization of Permutation Gates in the Third Level of the Clifford Hierarchy}, 
      author={Zhiyang He and Luke Robitaille and Xinyu Tan},
      year={2025},
      eprint={2510.04993},
      archivePrefix={arXiv},
      primaryClass={quant-ph},
      url={https://arxiv.org/abs/2510.04993}, 
}

@article{Beigi2010,
   author = {S. Beigi and P.W. Shor},
   doi = {10.26421/QIC10.1-2-4},
   issn = {15337146},
   issue = {1&2},
   journal = {Quantum Information and Computation},
   month = {1},
   pages = {41-59},
   title = {$\mathcal{C}_3$, semi-Clifford and genralized semi-Clifford operations},
   volume = {10},
   url = {http://www.rintonpress.com/journals/doi/QIC10.1-2-4.html},
   year = {2010},
}

@article{Cui2017,
   author = {Shawn X. Cui and Daniel Gottesman and Anirudh Krishna},
   doi = {10.1103/PhysRevA.95.012329},
   issn = {2469-9926},
   issue = {1},
   journal = {Physical Review A},
   month = {1},
   pages = {012329},
   title = {Diagonal gates in the Clifford hierarchy},
   volume = {95},
   url = {https://link.aps.org/doi/10.1103/PhysRevA.95.012329},
   year = {2017},
}

@article{Gottesman1999,
   author = {Daniel Gottesman and Isaac L. Chuang},
   doi = {10.1038/46503},
   issn = {0028-0836},
   issue = {6760},
   journal = {Nature},
   month = {11},
   pages = {390-393},
   title = {Demonstrating the viability of universal quantum computation using teleportation and single-qubit operations},
   volume = {402},
   url = {https://www.nature.com/articles/46503},
   year = {1999},
}

@article{Zeng2008,
   author = {Bei Zeng and Xie Chen and Isaac L. Chuang},
   doi = {10.1103/PhysRevA.77.042313},
   issn = {1050-2947},
   issue = {4},
   journal = {Physical Review A},
   month = {4},
   pages = {042313},
   title = {Semi-Clifford operations, structure of $\mathcal{C}_k$ hierarchy, and gate complexity for fault-tolerant quantum computation},
   volume = {77},
   url = {https://link.aps.org/doi/10.1103/PhysRevA.77.042313},
   year = {2008},
}

@article{Hu2021ClimbingHierarchy,
      title={Climbing the Diagonal Clifford Hierarchy}, 
      author={Jingzhen Hu and Qingzhong Liang and Robert Calderbank},
      year={2021},
      eprint={2110.11923},
      archivePrefix={arXiv},
      primaryClass={quant-ph},
      url={https://arxiv.org/abs/2110.11923}, 
}

@article{Jochym-OConnor2018DisjointnessGates,
    title = {{Disjointness of Stabilizer Codes and Limitations on Fault-Tolerant Logical Gates}},
    year = {2018},
    journal = {Physical Review X},
    author = {Jochym-O’Connor, Tomas and Kubica, Aleksander and Yoder, Theodore J.},
    number = {2},
    month = {5},
    pages = {021047},
    volume = {8},
    url = {https://link.aps.org/doi/10.1103/PhysRevX.8.021047},
    doi = {10.1103/PhysRevX.8.021047},
    issn = {2160-3308}
}

@article{deSilva2021EfficientDimensions,
    title = {{Efficient quantum gate teleportation in higher dimensions}},
    year = {2021},
    journal = {Proceedings of the Royal Society A: Mathematical, Physical and Engineering Sciences},
    author = {de Silva, Nadish},
    number = {2251},
    month = {7},
    pages = {20200865},
    volume = {477},
    url = {https://royalsocietypublishing.org/doi/10.1098/rspa.2020.0865},
    doi = {10.1098/rspa.2020.0865},
    issn = {1364-5021}
}

@article{deSilva2025,
      title={The Clifford hierarchy for one qubit or qudit}, 
      author={Nadish de Silva and Oscar Lautsch},
      year={2025},
      eprint={2501.07939},
      archivePrefix={arXiv},
      primaryClass={quant-ph},
      url={https://arxiv.org/abs/2501.07939}, 
}

\appendix
\section{Proof that all 1-wire-mismatch Permutations are in the Clifford Hierarchy}\label{App:1MM}
In this section we prove a general result about permutations in the Clifford Hierarchy. We use this result repeatedly when proving that certain families of permutation gates are in $\mathcal{CH}$.

First, note that we can write any permutation gate on $n$ qubits as a product of multi-controlled NOT gates. We allow the controls to be `ON' for either 1 or 0 which in standard gate notation is depicted as a closed or open circle, respectively. For a general permutation we must include NOT gates which are simply 0-controlled NOT gates, however, it is easy to show that they can always be pushed to either side of the product of gates by toggling controls and for our purposes it is not necessary to count them when calculating the wire mismatch. This representation is not unique and we only restrict the number of gates in the product to be finite.   

Now, we will define the wire mismatch of a permutation specified by a product of multi-controlled NOT gates as described above. 

\begin{Def}
    The \textbf{wire mismatch} of an $n$-qubit permutation is the number of wires (up to $n$) that have both targets and controls on them.
\end{Def}

For example, a 0-wire-mismatch permutation has only targets or controls on any wire. It follows that all gates in a 0-wire-mismatch permutation commute. It is easy to show that all 0-wire-mismatch permutations are in $\mathcal{CH}$ (see \cite{Anderson2024} for details).

\begin{Lem}\label{zeroWire}
    All 0-wire-mismatch permutations are in $\mathcal{CH}$.
\end{Lem}

In this section, we will prove that all 1-wire-mismatch permutations are in $\mathcal{CH}$. We will use the following lemmas:

From Anderson \cite{Anderson2024},
\begin{Lem}\label{gsc}
    For a permutation, $P$, and a diagonal gate $D$, $P\in \mathcal{CH}$ and $D\in \mathcal{CH}$ if and only if $ PD \in \mathcal{CH}$. 
\end{Lem}

From Zeng et al.~\cite{Zeng2008},
\begin{Lem}\label{LRcliff}
    A unitary (or gate), $U$, is in $\mathcal{CH}$ if and only if $C_L U C_R$ is in $\mathcal{CH}$ for all Clifford gates $C_L, C_R$.
\end{Lem}

From Cui et al.~\cite{Cui2017},
\begin{Lem}\label{diag}
For $D$ a diagonal gate and $P$ a permutation gate, $D\in \mathcal{CH}\implies PDP^{-1} \in \mathcal{CH}$ 
\end{Lem}
and
\begin{Lem}\label{diagGroup}
Diagonal gates in $\mathcal{CH}$ form a group. 
\end{Lem}
Now, we will prove the main result of this section.

\begin{theorem}
    All 1-wire-mismatch permutations are in $\mathcal{CH}$.
\end{theorem}
\begin{proof}
For a 1-wire-mismatch permutation, we will refer to the single wire that has both controls and targets as the mismatched wire. Note that every gate with a control on the mismatched wire has a target on some other wire.  We define the target support as the set of wires with targets from gates with a control on the mismatched wire. Define $P_0$ as the set of gates that have no support on the mismatched wire and no support (which would necessarily be a target) on the target support. Define $P_{targ}$ as the set of gates with no support on the mismatched wire, but with (target) support on a target support wire. W.l.o.g. we can write any 1-wire-mismatch as follows:
\begin{equation*}
    P_1 = \vec{X} P_0 P_{targ} P_{T_1} P_{C_1} P_{T_2} P_{C_2}\cdots P_{T_k}P_{C_k}. 
\end{equation*}
We have pulled all NOT (Pauli $X$) gates to the left by toggling controls. Gates in $P_0$ and $P_{targ}$ commute with all $P$ gates and can, therefore, be pulled to the left. The gates $P_{T_i}$ and $P_{C_i}$ are gates with support on the mismatched wire. These gates do not, in general, commute so we have grouped them in sets that share a target, $T_i$, or control, $C_i$, on the mismatched wire. We have written this expression with $P_{T_1}$ before $P_{C_1}$ in the product. This is not generally true, however, our proof is easily modified for the other case. 

Observe that all gates with targets on target support wires can be diagonalized by Hadamard gates on the target support wires. These diagonal gates individually have entries which are $\pm 1$ which Cui et al.~$\cite{Cui2017}$ showed are in $\mathcal{CH}$. Products of these diagonal gates are in $\mathcal{CH}$ by Lemma.~\ref{diagGroup} and we denote them by $D_{C_i}$. Then, we have that
\begin{equation*}
    P_1 = \vec{X} H_{targ} P_0 D_{targ} P_{T_1} D_{C_1} P_{T_2} D_{C_2}\cdots P_{T_k}D_{C_k} H_{targ}. 
\end{equation*}
Now, we proceed as follows:
\begin{equation*}
\begin{gathered}
    P_1 = \vec{X} H_{targ} P_0 P_{T_1}(P_{T_1}^{-1}D_{targ} P_{T_1} D_{C_1}) P_{T_2} D_{C_2}\cdots P_{T_k}D_{C_k} H_{targ}\\
    =  \vec{X} H_{targ} P_0 P_{T_1} P_{T_2} (P_{T_2}^{-1}(P_{T_1}^{-1}D_{targ} P_{T_1} D_{C_1}) P_{T_2} D_{C_2})\cdots P_{T_k}D_{C_k} H_{targ}\\
     \vdots\\
    =  \vec{X} H_{targ} P_0 P_{T_1} P_{T_2} \cdots P_{T_k} D_{[1\cdots k-1]}D_{C_k} H_{targ}.\\
\end{gathered}
\end{equation*}
The terms within the parenthesis indicate a diagonal gate. All terms $P_{T_i}^{-1}D_{C_j} P_{T_i}$ are in $\mathcal{CH}$ from Lemma.~\ref{diag} and because $D_{C_j}$ is in the Clifford Hierarchy. Additionally, products of diagonal gates in $\mathcal{CH}$ are in $\mathcal{CH}$ from Lemma.~\ref{diagGroup}. We indicate the product of these diagonal gates as: $D_{[1\cdots k-1]}$. With another application of Lemma.~\ref{diagGroup} we see that $$D = D_{[1\cdots k-1]}D_{C_k} \in \mathcal{CH}.$$ 

Recall that $P_{T_i}$ had mismatch only on the mismatched wire. But now, all permutations have only targets on the mismatched wire and we see that $P=P_0 P_{T_1} P_{T_2} \cdots P_{T_k}$ has 0-wire mismatch and is in $\mathcal{CH}$ by Lemma.~\ref{zeroWire}.

Finally, we have that :
\begin{equation*}
    P_1 = \vec{X}H_{targ}PDH_{targ}.
\end{equation*}

Since $\vec{X}$ and $H_{targ}$ are Clifford, we can use  Lemma.~\ref{LRcliff} to show that $$P_1 \in \mathcal{CH} \iff PD \in \mathcal{CH}.$$
and since $P\in \mathcal{CH}$ and $D\in \mathcal{CH}$ it follows from Lemma.~\ref{gsc} that $P_1 \in \mathcal{CH}$.
\end{proof}

\begin{corollary}\label{Corr:Fam}
    Any permutation, $P$, that can be written as: $$P = P_{C_L} P_{\le 1} P_{C_R}$$ where $P_{C_L}$ and $P_{C_R}$ are Clifford permutations and $P_{\le 1}$ is a general permutation with wire mismatch less than or equal to 1, is in the Clifford Hierarchy.
\end{corollary}

We note that there are known 2-wire-mismatch permutations which are not in the Clifford hierarchy (for example $CCX\cdot CXC$ is not in $\mathcal{CH}$ \cite{Zeng2008}). We also note that there are permutations with 2-wire-mismatch which are in the Clifford hierarchy (for example the ``Three-Toffoli Cascade'': $CCX_{1,2,3}\cdot CCX_{3,4,5} \cdot CCX_{5,6,7}$). Therefore, any extension beyond 1-wire-mismatch will have to take into account some additional structure of the wire mismatch. 

We introduce a proof-by-induction technique to prove that the $n$-Toffoli Cascade is in the Clifford hierarchy.

\subsection{Proof that the $n$-Toffoli Cascade is in the Clifford hierarchy}

\begin{Def}
    The $n$-Toffoli Cascade is a circuit consisting of $n$ Toffoli gates on $2(n-1)+3$ qubits defined as: $CCX_{1,2,3} \cdot CCX_{3,4,5} \cdots CCX_{2(n-1)+1, 2(n-1)+2, 2(n-1)+3}$.
\end{Def}
Again, we use Lemmas~\ref{gsc} \& \ref{LRcliff}. 

\begin{theorem}
For all integers $n \geq 1$, the $n$-Toffoli Cascade is in $\mathcal{CH}$.
\end{theorem}

\begin{proof}
We proceed by mathematical induction on $n$.

\paragraph{Base Case ($n = 1$):}
This is simply a proof that the Toffoli gate is in $\mathcal{CH}$. This can easily be verified and is shown in Zeng \etal~\cite{Zeng2008}.

\paragraph{Inductive Hypothesis:}
Assume that the $k$-Toffoli Cascade is in $\mathcal{CH}$ for some arbitrary integer $k\ge 1$.

\paragraph{Inductive Step:}
We must show that the statement also holds for $n = k + 1$.

First, note that the last Toffoli in the sequence, $CCX_{2k+1, 2k+2, 2k+3}$ (the rightmost Toffoli in a circuit diagram) has a target with no overlap with any other Toffoli gate. By conjugating this last wire by Hadamard we can turn this Toffoli into a CCZ gate. Since Hadamard is a Clifford gate, by Lemma.~\ref{LRcliff} this new circuit is in $\mathcal{CH}$ iff the original one was. This new circuit is a $k$-Toffoli Cascade multiplied by a single diagonal CCZ gate which is a generalized semi-Clifford gate. By Lemma.~\ref{gsc} this gate is in $\mathcal{CH}$ iff the $k$-Toffoli Cascade and the CCZ gate are in $\mathcal{CH}$. CCZ is known to be in $\mathcal{CH}$ at level 3 and a proof is Clifford equivalent to our base case above. By our inductive hypothesis, the $k$-Toffoli Cascade is in $\mathcal{CH}$. Therefore, the $(k+1)$-Toffoli gate is in $\mathcal{CH}$ which completes our proof. 
\end{proof}

\section{Extensions of Affine Equivalence}\label{affineproof}

In this section, we investigate extensions to the notion of affine equivalence (AE) of permutations used in the main text. This is important since the number of affine equivalence classes increases rapidly with the number of qubits, $n$ (see Table~\ref{table:numAE}), rendering it impractical to search over all AE classes even for small $n$. For questions such as: {\it Is a permutation (generalized) semi-Clifford?} or {\it Is a permutation in the $k$th level of the Clifford hierarchy?}, the most general notion of equivalence is that the permutations are equal up to left and right multiplication by Clifford gates. Henceforth, we will refer to this as Clifford equivalence (CE). If CE of permutations was found to be much richer than AE (and an algorithm of similar complexity was found for CE testing), then it could be used to investigate properties of permutations for $n \ge 5$. This section shows two results: (1) We find that extending AE to monomial Clifford equivalence adds nothing and (2) we find permutations which are CE, but not AE. Initial investigations indicate that CE is not vastly richer than AE, but further investigations are warranted.    

\subsection{Monomial Clifford Equivalence of Permutations}

Two permutations, $P_1$ and $P_2$, are affine equivalent if and only if there exists Clifford permutations, $C_{\pi_1}$ and $C_{\pi_2}$, such that $P_1 = C_{\pi_1}P_2C_{\pi_2}$. Note that since permutations form a group the product $C_{\pi_1}P_2C_{\pi_2}$ is a permutation for any choice of $C_{\pi_1}, C_{\pi_2}$. We denote this equivalence relation as: $\equiv_{AE}$.

We define monomial Clifford equivalence (MCE) similarly. Now, two permutations, $P_1$ and $P_2$, are monomial Clifford equivalent if and only if there exists Clifford monomial gates, $C_{M_1}=C_{\pi_1}C_{d_1}$ and $C_{M_2}=C_{\pi_2}C_{d_2}$, such that $P_1 = C_{M_1}P_2C_{M_2}$. Here $C_{d_i}$ denotes a diagonal Clifford gate. While the Clifford monomial gates form a group, the product $C_{M_1}P_2C_{M_2}$ is {\it not} a permutation for all choices of $C_{M_1}, C_{M_2}$. We denote this equivalence relation as: $\equiv_{MCE}$. Note that MCE is at least as powerful as AE since setting $C_{d_i}=I$ is the definition of AE.  

Here we prove the following theorem:

\begin{theorem}
    Two permutations are affine equivalent if and only if they are monomial Clifford equivalent.
\end{theorem}

\begin{proof}
    ($\implies$) Since the monomial Clifford gates contain all Clifford (affine) permutations, any affine equivalence implies a monomial Clifford equivalence. 

    ($\impliedby$) Now, we will prove that MCE implies AE when applied to permutations. We prove this by examining the conditions under which $C_{M_a} P_2 C_{M_b}$ is a permutation. This is a necessary condition for $P_1 \equiv_{MCE} P_2$.

    First, note that any monomial Clifford gate, $C_M$, can be written as a product $C_{\pi}C_{d}$ where $C_{\pi}$ is a Clifford permutation and $C_d$ is a diagonal Clifford gate. Then, 
    \begin{align*}
        P_1 &\equiv_{MCE} P_2 \iff P_1 = C_{M_a} P_2 C_{M_b}  \\&= C_{\pi_a}C_{d_a} P_2 C_{\pi_b}C_{d_b} = C_{\pi_a} P_2 (P_2^{-1} C_{d_a} P_2) C_{\pi_b}C_{d_b} \\&= (C_{\pi_a} P_2 C_{\pi_b}) (C_{\pi_b}^{-1} P_2^{-1} C_{d_a} P_2 C_{\pi_b} C_{d_b}).
    \end{align*}

We see that the terms in the first parenthesis are all permutations and therefore the product is a permutation. Since a diagonal matrix conjugated by a permutation matrix is always a diagonal matrix, we conclude that the product, $C_{\pi_b}^{-1} P_2^{-1} C_{d_a} P_2 C_{\pi_b} C_{d_b}$, is a diagonal matrix. For $C_{M_a} P_2 C_{M_b}$ to be a permutation, the diagonal part must equal Identity since this is the only element in the union of diagonal gates and permutations. 

We can express this constraint as a constraint on $C_{d_b}$ as follows: $$C_{\pi_b}^{-1} P_2^{-1} C_{d_a} P_2 C_{\pi_b} C_{d_b} = I \implies C_{d_b} = C_{\pi_b}^{-1} P_2^{-1}C_{d_a}^{-1} P_2 C_{\pi_b}.$$

Finally, we can substitute this constraint into our equivalence relation:  
\begin{align*}
    P_1 &\equiv_{MCE} P_2 \\&\implies P_1 = C_{\pi_a}C_{d_a} P_2 C_{\pi_b}C_{d_b} \\&\implies P_1 = C_{\pi_a} C_{d_a} P_2 C_{\pi_b} (C_{\pi_b}^{-1} P_2^{-1}C_{d_a}^{-1} P_2 C_{\pi_b}) \\&\implies P_1 = C_{\pi_a} C_{d_a} P_2 P_2^{-1}C_{d_a}^{-1} P_2 C_{\pi_b} \\&\implies P_1 = C_{\pi_a} P_2 C_{\pi_b} \\&\implies P_1 \equiv_{AE} P_2.
\end{align*}
\end{proof}

\subsection{Clifford Equivalence is not equal to Affine Equivalence}

We find the following 5-qubit permutations are Clifford equivalent (as can be seen by conjugating the first, third, and fifth wires by Hadamard):
\begin{figure}[!h]
\begin{equation*}
\begin{tabular}{m{0.0\textwidth} m{0.0\textwidth} m{0.01\textwidth} m{0.01\textwidth}}
    \centering 
    \Qcircuit @C=0.3em @R=0.8em {
    & \ctrl{1} & \ctrl{3} & \qw \\ 
    & \ctrl{1} & \qw      & \qw \\
    & \targ    & \qw      & \qw \\
    & \qw      & \ctrl{1} & \qw \\
    & \qw      & \targ    & \qw  \\
    } & , &
    \centering
    \Qcircuit @C=0.3em @R=1.0em {
    & \targ & \targ & \qw \\ 
    & \ctrl{-1} & \qw      & \qw \\
    & \ctrl{-1} \qw    & \qw      & \qw \\
    & \qw      & \ctrl{-3} & \qw \\
    & \qw      & \ctrl{-1} & \qw  \\
    } & . \\
\end{tabular}
\end{equation*}
\end{figure}

We then check these permutations for AE with algorithms from \cite{Draper2009, Canniere2007} and find that they are not AE. We conclude that CE and AE are not equal and (by the pigeonhole principle) that some distinct affine equivalence classes must be equivalent under CE.   

\section{Properties of Diagonal Gate Groups in $\mathcal{CH}$} 

\label{sec:diag}
In this section, we give further properties of two diagonal gate groups in the Clifford Hierarchy. The interplay between these groups was used extensively in sections~\ref{sec:atk34} and \ref{sec:atk35} to reduce the number of equivalence classes that had to be checked in our proof. Much of this discussion including Fig.~\ref{fig:diagGroups} is from \cite{Anderson2024}. 

First, we will introduce the group, $\mathcal{D}_k^n$. This is the group of diagonal gates on $n$ qubits ($2^n \times2^n$ matrices) in $\mathcal{CH}_k$. We will define this group by a set of $Z$ rotation generators which (we think) makes it easy to see what level in $\mathcal{CH}$ that a product of generators is in. We can write a diagonal element, $d$, as a product of $Z$ rotations as follows:

\begin{equation}\label{diagForm}
    d=\prod_j\exp \left( i\frac{\alpha_j \pi}{2^{k_j}}Z_j \right).
\end{equation}
Here $\alpha_j$ is an integer, $k_j$ is a positive integer, and $Z_j$ is a Pauli $Z$ string. A diagonal gate is in $\mathcal{CH}$ iff it can be written this way and $d$ is in the $k=\max k_j$ level in the Hierarchy. In Cui \etal\cite{Cui2017}, they showed that all $n$-qubit $Z$ rotations by a fixed angle $\pi/2^k$ generate the group $\mathcal{D}_k^n$. Where this group sits in the Clifford Hierarchy is most easily seen in Fig.~\ref{fig:diagGroups}. By definition, products of elements in $\mathcal{D}_k^n$ are in the same level in $\mathcal{CH}$. However, in this work we often conjugate a diagonal element by non-Clifford permutations. This does not necessarily preserve the level in $\mathcal{CH}$ and can, therefore, map elements out of $\mathcal{D}_k^n$. For this reason we introduce a second diagonal gate group.

While conjugation by permutations does not necessarily preserve the level of a diagonal gate, it does preserve the spectrum (entries of the diagonal gate). We introduce the following group which is preserved by conjugation with permutations. We define the diagonal matrix group: $Diag^n_r$. This is the group of all diagonal $2^n \times 2^n$ matrices with entries in the $2^r$th roots-of-unity. From \cite{Cui2017} these diagonal matrices are all in the Clifford Hierarchy. 

It is useful to define these groups by a generating set of controlled diagonal gates  Any diagonal gate in $\mathcal{CH}$ on $n$ qubits can be written (up to a global phase) as some combination of gates:

\begin{equation*}
    \Lambda_{n-1}(Z^{1/2^{k_j}}), \Lambda_{n-2}(Z^{1/2^{k_j}}), \cdots , \Lambda_{1}(Z^{1/2^{k_j}}), Z^{1/2^{k_j}}.
\end{equation*}

With 
\begin{equation*}
Z^{1/2^{k_j}} \triangleq \left [ \begin{array}{lc} 1 & 0 \\ 0 & e^{i\pi/2^{k_j}}\\ \end{array} \right ].    
\end{equation*}

A gate $\Lambda_{m}(Z^{1/2^{l}})$ is in $\mathcal{CH}_{l+m}$ \cite{Hu2021ClimbingHierarchy,Cui2017}. Therefore, all gates on $n$ qubits with $l+m\le k$ form the group $\mathcal{D}_k^n$. 

We can also express the group, $Diag_r^n$, with the gates. For fixed $r$ and $n$, $Diag_r^n$ is generated by:

\begin{equation*}
    \langle \Lambda_{n-1}(Z^{1/2^r}), \Lambda_{n-2}(Z^{1/2^r}), \cdots , \Lambda_{1}(Z^{1/2^r}), Z^{1/2^r}\rangle.
\end{equation*}

We have the following relation between groups: $\mathcal{D}_{k}^n \subseteq Diag_{k}^n \subseteq \mathcal{D}_{k+n}^n \subseteq Diag_{k+n}^n$. This relation is visualized in Fig.~\ref{fig:diagGroups}. It follows that $\mathcal{D}_{k_2}^n$ can always be chosen to be large enough to contain all elements of $Diag_{k_1}^n$. Therefore, both $Diag_k^n$ and $\mathcal{D}_k^n$ contain any diagonal element in $\mathcal{CH}$ given sufficiently large $k$. 

\begin{figure}[!h]
     \centering
     \begin{subfigure}
         \centering
         \includegraphics[scale=0.44]{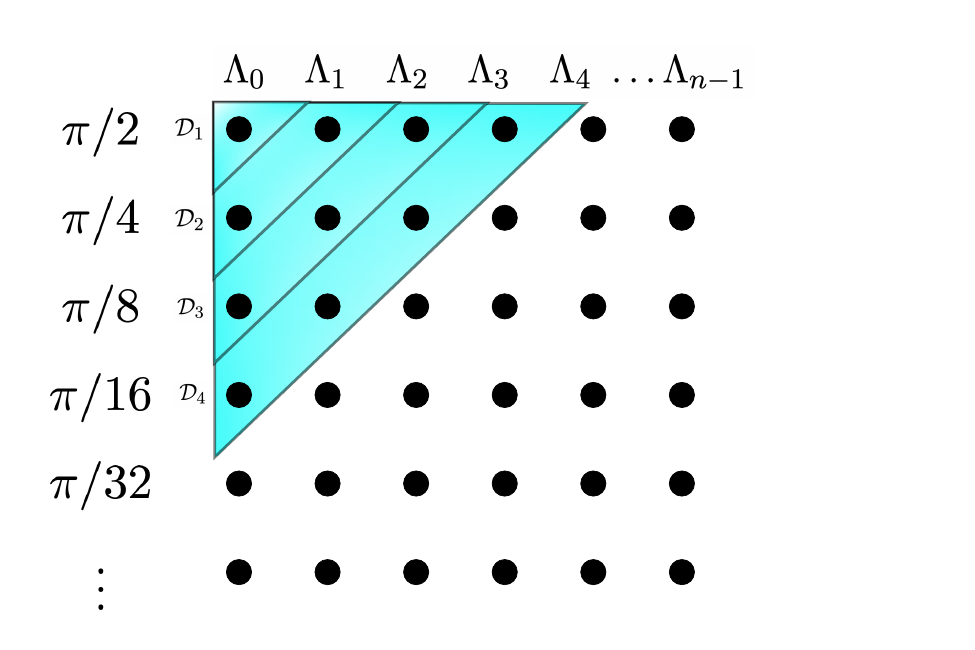}
     \end{subfigure}
     \hfill
     \begin{subfigure}
         \centering
         \includegraphics[scale=0.44]{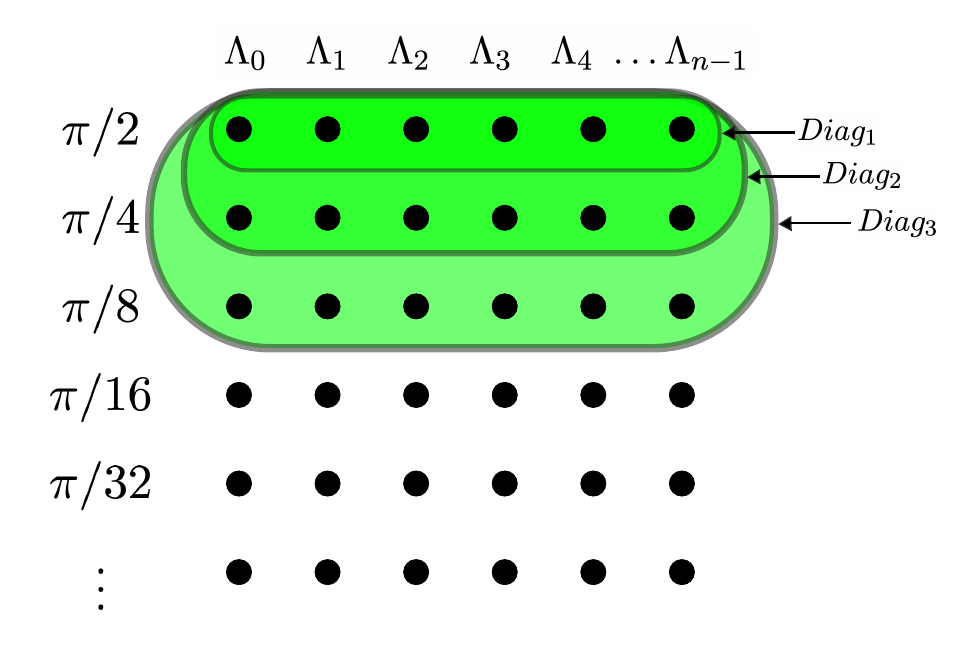}
     \end{subfigure}
        \caption{Comparison of diagonal gate groups in $\mathcal{CH}$. Points in figures correspond to controlled-$Z$ rotations by angle $\pi/2^k$. The number of controls is labeled across the top and the rotation angle is on the left. The groups are generated by these multi-controlled rotations. The figure on the left shows the family of groups, $\mathcal{D}^n_k$, which are the diagonal gates at level $k$ in $\mathcal{CH}$. The figure on the right shows the groups, $Diag^n_k$, of matrices with entries in the $2^k$th root of unity. Here the gates act on up to $n$ qubits and the generators are on all appropriate sets of qubits. For example, the $\Lambda_2(Z\left[\frac{\pi}{2}\right])$ ($CCZ$ gate) generators are on all $\binom{n}{3}$ qubits.}
        \label{fig:diagGroups}
\end{figure}

\end{document}